\documentclass[10pt, twocolumn]{IEEEtran}
\IEEEoverridecommandlockouts 

\def\withcolors{0}
\def\withnotes{0}

\usepackage[usenames, dvipsnames, table]{xcolor}
\usepackage[noend]{algpseudocode}
\usepackage[cmex10]{amsmath}
\usepackage[mathscr]{eucal}
\usepackage[noadjust]{cite}
\usepackage[colorinlistoftodos, textsize=scriptsize]{todonotes}
\usepackage[normalem]{ulem} 
\usepackage{%
	amsfonts,%
    	amssymb,%
    	amsthm,
    	bbm,%
	circuitikz,%
	enumitem,%
    	etoolbox,%
    	mathtools,%
	pgf,%
	pgfplots,%
    	stmaryrd,%
	subcaption,%
	tikz,%
	url,%
}

\newlist{steps}{enumerate}{1}
\setlist[steps, 1]{label = \quad \quad \quad Step \arabic*:}
\usepgflibrary{shapes}
\usetikzlibrary{%
    arrows,%
    decorations.text,%
    positioning,%
    scopes,%
    shapes,
    patterns,%
}	

\theoremstyle{plain} \newtheorem{thm}{Theorem}[section]
\newtheorem{lem}[thm]{Lemma} 
\newtheorem{cor}[thm]{Corollary}

\theoremstyle{definition} \newtheorem{defn}[thm]{Definition}

\newtheorem{example}[thm]{Example}

\theoremstyle{remark} \newtheorem{rem}{Remark}

\definecolor{lightgray}{gray}{0.9}

\newcommand{\ep}{\varepsilon}
\newcommand{\eq}[1]{\begin{align*}#1\end{align*}}
\newcommand{\ceil}[1]{\left\lceil#1\right\rceil}%
\newcommand{\norm}[1]{\left\lVert#1\right\rVert}%

 \newcommand{\R}{\mathbb{R}}
\newcommand{\N}{\mathbb{N}}
\newcommand{\E}[1]{\mathbb{E}\left[#1\right]}

 \newcommand{\G}{\mathcal{G}}
\newcommand{\X}{\mathcal{X}}
\newcommand{\Y}{\mathcal{Y}}
\newcommand{\uRoman}[1]{\uppercase\expandafter{\romannumeral#1}}

\DeclarePairedDelimiter\floor{\lfloor}{\rfloor}

\ifnum\withcolors=1
  \newcommand{\newer}[1]{{\color{blue} {#1}}} 
  \newcommand{\mcolor}[1]{{\color{ForestGreen}#1}} 
  \newcommand{\pcolor}[1]{{\color{RubineRed}#1}} 
  \newcommand{\tcolor}[1]{{\color{Orange}#1}} 
\else
  
  \newcommand{\newer}[1]{{{#1}}}
  \newcommand{\mcolor}[1]{{#1}}
  \newcommand{\pcolor}[1]{{#1}}
  \newcommand{\tcolor}[1]{{#1}}
\fi

\ifnum\withnotes=1
  \newcommand{\mnote}[1]{\par\mcolor{\textbf{M: }\sf #1}} 
  \newcommand{\pnote}[1]{\par\pcolor{\textbf{P: }\sf #1}} 
  \newcommand{\tnote}[1]{\par\tcolor{\textbf{T: }\sf #1}} 
  
\else
  \newcommand{\mnote}[1]{}
  \newcommand{\pnote}[1]{}
  \newcommand{\tnote}[1]{}
  
\fi
\newcommand{\ignore}[1]{\leavevmode\unskip} 

\begin{document}
\title{Optimal Source Codes for Timely Updates}
\author{\IEEEauthorblockN{Prathamesh Mayekar,~\IEEEmembership{Student Member,~IEEE} } 
\and
  \IEEEauthorblockN{Parimal Parag,~\IEEEmembership{Member,~IEEE}} 
  \and
  \IEEEauthorblockN{Himanshu Tyagi,~\IEEEmembership{Senior
      Member,~IEEE}}}

    \date{} \maketitle

{\renewcommand{\thefootnote}{}\footnotetext{
The authors are with the Department of Electrical Communication Engineering,
Indian Institute of Science, Bangalore 560012, India.  (Email: \{prathamesh, parimal, htyagi\}@iisc.ac.in).

A preliminary  version of this paper~\cite{mayekar2018optimal} was presented at the IEEE International
    Symposium on Information Theory, Vail, USA, 2018.
}
}

\maketitle \renewcommand{\thefootnote}{\arabic{footnote}}
\setcounter{footnote}{0}

\begin{abstract}
A transmitter observing a sequence of independent and identically
distributed random variables seeks to keep a receiver updated about
its latest observations.  The receiver need not be apprised about each
symbol seen by the transmitter, but needs to output a symbol at each
time instant $t$.  If at time $t$ the receiver outputs the symbol seen
by the transmitter at time $U(t)\leq t$, the age of information at the
receiver at time $t$ is $t-U(t)$.  We study the design of lossless
source codes that enable transmission with minimum average age at the
receiver.  We show that the asymptotic minimum average age can be
attained up to a constant gap by the Shannon codes for a tilted version of the original
pmf generating the symbols, which can be computed easily by solving an
optimization problem.  Furthermore, we exhibit an example with
alphabet $\X$ \newer{where  Shannon} codes for the original pmf incur an
asymptotic average age of a factor $O(\sqrt{\log |\X|})$ more
than that achieved by our codes.  Underlying our prescription for
optimal codes is a new variational formula for integer moments of
random variables, which may be of independent interest.  Also, we
discuss possible extensions of our formulation to randomized schemes
\newer{and to the erasure channel,} and include a treatment of the related problem of
source coding for minimum average queuing delay.
\end{abstract} 




\section{Introduction}
Timeliness is emerging as an important requirement for communication
{in} cyber-physical systems (CPS). {Broadly}, it refers to the
requirement of having the latest information from the transmitter
available at the receiver in a timely fashion.  It is important to
distinguish the requirement of timeliness from that of low delay
transmission: The latter places a constraint on the delay in
transmission of each message, while timeliness is concerned about how
recent is the current information at the receiver.  
In particular, the instantaneous staleness at the receiver is low if a message is received with low delay. 
However, the instantaneous staleness increases linearly at the receiver until a subsequent message is decoded successfully.
{A heuristically appealing metric that can capture the notion of timeliness of
  information in a variety of applications, termed its {\em age}, was first used
  in~\cite{KaulYatesGruteser11} for a setting involving queuing and
  link layer delays and was analyzed systematically for a queuing
  model in the pioneering work~\cite{KaulYatesGruteser12}; see~\cite{yates2015lazy,
    he2018optimal, sun2017update, bhambay2017differential, 
    bacinoglu2017scheduling, kosta2017age} for a sampling of
  subsequent developments in problems related to minimum age
  scheduling.} In this paper, we initiate a systematic study of the
design of source codes with the goal of minimizing the age of the
information at the receiver.

As a motivating application, consider 
  remote sensor data monitoring where
  at each instant the sensor observes real-valued, time-series measurements. For concreteness, the reader may consider voltage and current data recording using \newer{intelligent electronic devices} in a power distribution network. 
  The sensor communicates to a center over a network to enable fault detection and
  fault analysis. On the one hand, the communication protocol and buffer constraints at the sensor limits the rate at which the sensor can send data packets to the center. On the other hand, it is not very important for the center to get all the packets from the sensor.  
  Rather the center wants timely updates about the sensor observations.
  In fact, when operating with cheap hardware with limited front-end buffers, it is common to have observation values in the buffer
  overwritten as new recordings are made even before the previous one waiting in the buffer has been picked-up for processing. Our work focuses on data compression for such applications where there is no direct cost of skipping packets and the interest is only in timely updates.

Specifically, we consider the problem of source coding where a
transmitter receives symbols generated from a known distribution and
seeks to communicate them to a receiver in a timely fashion.\footnote{
This assumption of known distribution is realized in practice by building a model for sensor data offline, before initiating the live monitoring process.}
\newer{To that
end, it encodes a symbol $x$ to $e(x)$ using a variable length
prefix-free code $e$. } The coded sequence is then transmitted over a
noiseless communication channel that sends one bit per unit time.  We
restrict our treatment to a simple class of deterministic\footnote{Our
  analysis of average age extends to randomized schemes as well; see
  Section~\ref{s:extensions}.} update schemes, termed {\it memoryless
  update schemes}, where the transmitter does not have have a buffer to store the symbols it has seen previously and simply sends the next observed symbol once the 
  channel is free.

Specifically, denoting the source alphabet by $\X$,
the transmitter observes a symbol $X_t \in \X$ at each discrete time $t$. 
At time $t=1$, 
the transmitter communicates the symbol $X_1 = x_1$ by encoding it to codeword $e(x_1)$ of length $\ell(x_1)$ bits. 
This transmission requires $\ell(x_1)$ channel uses and is received perfectly at the decoder at time $1+\ell(x_1)$.  
Since the channel remains busy sending $e(x_1)$ for time instants $1$ to $\ell(x_1)$, the transmitter cannot send any new symbols during this period. 
At time $t' = 1+\ell(x_1)$, the transmitter observes the symbol $X_{t'} = x_{t'}$.
\newer{Under a memoryless update scheme, }
the transmitter cannot store the symbols seen during the time interval $\{2, \dots, \ell(x)\}$ and communicates codeword $e(x_{t'})$ over the next $\ell(x_{t'})$ channel uses, starting from the time instant $t'=1+\ell(x_1)$.
The communication process continues repeatedly in this fashion.

We emphasize that under memoryless schemes, the source symbols
generated and observed by the transmitter while the channel is busy
sending a previous symbol are simply skipped. This skipping is only
allowed when the channel is busy, and not at the will of the encoder
when the channel is free (see Section~\ref{s:extensions} for
discussion on randomized schemes that allow the transmitted to skip
symbols even when channel is free). Furthermore, the encoder need not
indicate to the decoder that a symbol has been skipped using a special
symbol -- the decoder can ascertain this from the received
communication since the  
channel is noiseless  \newer{and compression is done using prefix-free codes}.  


On the receiver side, at each instance $t$ the decoder outputs a time
$U(t)$ and the symbol $X_{U(t)}$ seen by the transmitter at time
$U(t)$.  Thus, the {\em age of information} at the receiver at time $t$ is
given by $A(t) = t - U(t)$.  
  We note that age of information measures timeliness at the receiver.
  When the transmitter skips source symbols, $U(t)$ remains unchanged at the receiver and the age $A(t)$ increases. Therefore, the age metric implicitly penalizes for
  skipping symbols.

We illustrate the setup in Figure~\ref{f:setup}. 
 In this example,
 the symbol $X_1$ generated at time $t=1$ is encoded to a two-bit codeword $e(X_1)$ and received at the decoder at time $t=3$ after two channel uses.
At time $t=2$, the transmitter skips symbol $X_2$  since the channel was busy sending $X_1$ when it arrived. Further, the decoder retains $U(t)=0$ since it has not received any symbol. 
 At time $t=3$, the decoder receives the codeword $e(X_1)$, updates $U(3) = 1$, and outputs the corresponding  symbol $X_1$. Thus, the age of information at the receiver at time $t = 3$ is $A(3) = 2$. Since the channel becomes available at time $t=3$, the transmitter encodes the symbol $X_3$ and transmits the one-bit codeword $e(X_3)$, which is received after a single channel-use at time $t=4$. 
At time $t=4$, the decoder outputs time $U(4) = 3$ with outputs the corresponding symbol $X_3$, 
and the age of information at the receiver is $A(4) = 1$. Once again, the channel becomes available at time $t=4$ for the transmitter. It encodes the current symbol $X_4$ into the codeword $e(X_4)$ of length $3$ bits and sends $e(X_4)$ over the channel; $e(X_4)$ is received at time $t=7$. 
The decoder retains the output $U(t) = 3$ and $X_{U(t)} = X_3$ for times $t \in \{4,5,6\}$. At time $t=7$, the decoder outputs time $U(7) = 4$ and the corresponding symbol $X_4$; the age of information at the receiver is $A(7) = 3$. 

\begin{figure}[h]
\centering \begin{tikzpicture}
[scale=1, transform shape,
    pre/.style={=stealth',semithick},
    post/.style={->,shorten >=1pt,>=stealth',semithick},
dimarrow/.style={->, >=latex, line width=1pt}
    ]
\draw[<-,line width=1pt] (0.25,3.5) to (0.25,0);
\node[align=center, rotate=90] at (0.0,3.0) {Time}; 

\draw [dashed, color=black!20] (0.3, 0.0) to (8.5,0.0);
\node[align=center] at (0.6,0.2) {t=1}; 
\draw [dashed, color=black!20] (0.3, 0.5) to (4.5, 0.5);
\node[align=center] at (0.6,0.7) {t=2}; 
 \draw [dashed, color=black!20] (5.5, 0.5) to (8.5,0.5);
 \draw [dashed, color=black!20] (0.3, 2*0.5) to (8.5,2*0.5);
\node[align=center] at (0.6,0.7+0.5) {t=3}; 
\draw [dashed, color=black!20] (0.3, 3*0.5) to (8.5,3*0.5);
\node[align=center] at (0.6,0.7+2*0.5) {t=4}; 
\draw [dashed, color=black!20] (0.3, 4*0.5) to (4.5, 4*0.5);
\node[align=center] at (0.6,0.7+3*0.5) {t=5};
\draw [dashed, color=black!20] (5.5, 4*0.5) to(7.5,4*0.5);
\draw [dashed, color=black!20] (0.3, 5*0.5) to (4.5, 5*0.5);
\node[align=center] at (0.6,0.7+4*0.5) {t=6}; 
\draw [dashed, color=black!20] (5.5, 5*0.5) to(7.5,5*0.5);
\draw [dashed, color=black!20] (0.3, 6*0.5) to (8.5,6*0.5);
\node[align=center] at (0.6,0.7+5*0.5) {t=7};

\draw[thick, rounded corners=1.mm] (0.25, -0.25) rectangle (1.75, -.75);  
\node[align=center] at (1.02,-0.5) {Source};  
\draw[dimarrow] (1.75,-0.5) to (2.5,-0.5);
\draw[thick, rounded corners=1.mm] (2.5, -0.25) rectangle (4, -.75);
\node[align=center] at (3.25,-0.5) {Encoder};
\draw[dimarrow] (4,-0.5) to (4.3,-0.5);
\draw[thick, rounded corners=1.mm] (4.3,-0.25) rectangle (5.8, -0.75);
\node[align=center] at (5.1,-0.5) {Channel };
\draw[dimarrow] (5.8,-0.5) to (6.1,-0.5);
\draw[thick, rounded corners=1.mm] (6.1, -0.250) rectangle (7.6, -0.75);
\node[align=center] at (6.9,-0.5) {Decoder};
\draw[dimarrow] (7.6,-0.5) to (8.1,-0.5);
 
\draw[pattern=vertical lines, pattern color=blue!60] (1.5,0) rectangle (2.5,0.5); 
\node[align=center] at (2.1,0.25) {$\bold{X_1}$ };

\draw[pattern=horizontal lines, pattern color=cyan!60] (1.5,0.5) rectangle (2.5,1); 
\node[align=center] at (2.1,0.75) {$\bold{X_2}$ };

\draw[pattern=north east lines, pattern color=gray!60] (1.5,1) rectangle (2.5,1.5); 
\node[align=center] at (2.1,1.25) {$\bold{X_3}$ };

\draw[pattern=north west lines, pattern color=red!20] (1.5,1.5) rectangle (2.5,2); 
\node[align=center] at (2.1,1.75) {$\bold{X_4}$ };


\draw[ color=black, fill = black!75] (2,2.15) circle (0.05cm); 
\draw[ color=black, fill = black!75] (2,2.35) circle (0.05cm); 
\draw[ color=black, fill = black!75] (2,2.55) circle (0.05cm); 

\draw[pattern=vertical lines, pattern color=blue!60] (4.5,0) rectangle (5.5,1); 

\node[align=center] at (5.1,0.25) {$\bold{e(X_1)}$ };

\draw[pattern=north east lines, pattern color=gray!60] (4.5,1) rectangle (5.5,1.5); 
\node[align=center] at (5.1,1.25) {$\bold{e(X_3)}$ };

\draw[pattern=north west lines, pattern color=red!20] (4.5,1.5) rectangle (5.5,3); 
\node[align=center] at (5.1,1.75) {$\bold{e(X_4)}$ };

\draw[ color=black, fill = black!75] (5,3.15) circle (0.05cm); 
\draw[ color=black, fill = black!75] (5,3.35) circle (0.05cm); 
\draw[ color=black, fill = black!75] (5,3.55) circle (0.05cm);


\draw[pattern= vertical lines, pattern color=blue!60] (7.5,1) rectangle (8.5,1.5); 
\node[align=center] at (8.1,1.25) {$\bold{X_1}$ };

\draw[pattern=north east lines, pattern color=gray!60] (7.5,1.5) rectangle (8.5,3); 
\node[align=center] at (8.1,1.75) {$\bold{X_3}$ };

\draw[pattern=north west lines, pattern color=red!20] (7.5,1.5+3*0.5) rectangle (8.5,2+3*0.5); 
\node[align=center] at (8.1, 1.75+3*0.5)  {$\bold{X_4}$};
\draw[ color=black, fill = black!75] (8,3.15+0.5) circle (0.05cm); 
\draw[ color=black, fill = black!75] (8,3.35+0.5) circle (0.05cm); 
\draw[ color=black, fill = black!75] (8,3.55+0.5) circle (0.05cm);
\end{tikzpicture}
\caption{Illustration of a memoryless update scheme for the first 4 packets in the source-queue.}
\label{f:setup}
\end{figure}
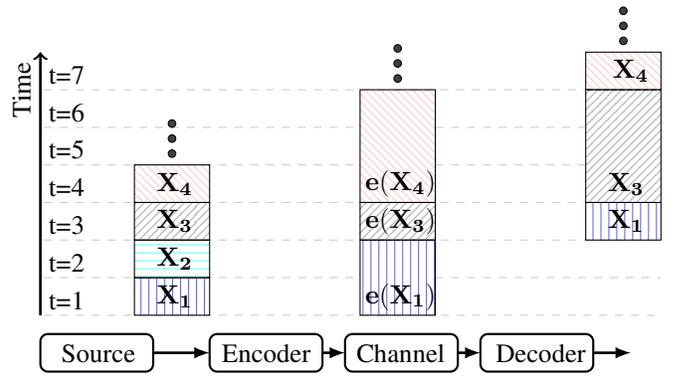

Our goal in this paper is to design prefix-free codes for which the
average age of the memoryless scheme above is minimized; namely codes
$e$ that minimize
\[
\bar{A}(e) = \lim_{T\to \infty} \frac 1T \sum_{t=1}^T A(t).
\]
This formulation is apt for the timely update problem where the
transmitter need not send each update and strives only to reduce the
average age of the information at the receiver.

Using a simple extension of the renewal reward
  theorem, we derive a closed form formula for the asymptotic average
age attained by a prefix-free code.  Interestingly, this formula is a
rational function of the first and the second moment of the random
codeword length.  Our main technical contribution in this paper is a
variational formula for the second moment of random variables that
enables an algorithm for finding the code that attains the minimum
asymptotic average age up to a constant gap.  The variational formula
is of independent interest and may be useful in other settings where
such cost functions arise; we point-out one such setting in
Section~\ref{s:extensions}. In fact, our prescribed prefix-free code
is a Shannon code\footnote{A Shannon
code for $P$ is a prefix-free code that assigns lengths $\ell_S(x) = \lceil -\log P(x)\rceil$ to
a symbol $x$ ($cf.$~\cite{CovTho06}).} for a tilted version of the original pmf. See \eqref{e:tilted_distribution} below for the \newer{description} of the tilted version; it can be computed by solving an optimization problem entailing entropy maximization.

The formula for average age that we derive yields an $O(\log |\X|)$
upper bound on the minimum average age, attained by a fixed length
code.  We show that the same upper bound of $O(\log|\X|)$ holds for
the average age of a Shannon code for the original distribution as
well. However, we exhibit an example where Shannon codes for the original
distribution have $\Omega(\log|\X|)$ age, while our aforementioned
proposed code 
yields an average age of $O(\sqrt{\log |\X|})$.

The problem of designing update codes with low average age is related
to real-time source coding ($cf.$~\cite{Mahajan}) where we seek to
transmit a stream of data under strict delay bounds.  A related
formulation has emerged in the control over communication network
literature ($cf.$~\cite{TatikondaMitter04}) where an observation is
quantized and sent to an estimator/controller to enable control.
Here, too, the requirement is that of communication under bounded
delay.

An alternative formulation for minimum age source coding is considered
in the recent work~\cite{YatesDCC}. Unlike our formulation,
skipping of symbols is prohibited
in~\cite{YatesDCC}. Instead, the authors consider
fixed-to-variable length block codes and require that each coded
symbol be transmitted over a constant rate, noiseless bit-pipe. In
this setting, an exact expression for average age is not available,
and the authors take recourse to an approximation for average age. This
approximate average age is then optimized numerically over a set of
prefix-free codes using the algorithm in
\cite{larmore1989minimum}. The authors further reduce the computational complexity of this algorithm by using  the algorithm in \cite{baer2006source}.

A recent paper~\cite{ZhongYatesSoljanin18} extends this problem to include random arrival times of source symbols and applies the algorithm from~\cite{larmore1989minimum} for optimizing the cost function.
 Note that the cost function optimized
  in \cite{larmore1989minimum} is similar to the approximate average
  age of~\cite{YatesDCC , ZhongYatesSoljanin18}, but with one crucial difference:  While the former is monotonic in both first and second moments of
  random lengths, the latter is not. In absence of this monotonicity,
  the optimality of the solution produced by algorithm in
  \cite{larmore1989minimum} is not guaranteed for the cost functions in~\cite{YatesDCC, ZhongYatesSoljanin18}. In a related
work~\cite{zhong2017backlog}, the same authors point-out that the
average age can be further reduced by allowing the encoder to
dynamically control the block-length of the fixed-to-variable length
codes.

\newer{In contrast to~\cite{YatesDCC}, which is perhaps closest to our work, 
  we derive an exact expression for average age and rigorously establish the structural properties of the optimal solution to the relaxed problem.}
In fact, our proposed minimum average age problem differs from all these prior
formulations since we need not send the entire stream and are allowed
to skip some symbols. 
In our applications of interest, such as that of real-time sensor data monitoring outlined earlier, the allowed
communication rates are much lower than the rate at which data is
generated. Thus, there is no hope of transmitting all the data at
bounded delay, as mandated by the formulations available hitherto. 
Nonetheless, our setting is related closely to that
  in~\cite{YatesDCC} and provides a complementary
  formulation for age optimal source coding.
    We note that our focus is on settings where the alphabet size of the streaming symbols is large. In such settings, the average age for any memoryless update scheme would be much larger than a small constant. Therefore, it suffices to establish optimality up to small additive constants.

In addition to our basic formulation, we present a few extensions of
our formulations and other use cases for our proposed variational
formula. Specifically, while we restrict to deterministic schemes for
the most part, our analysis can be extended easily to analyze
randomized schemes where the encoder can choose to skip an available
transmission slot randomly. This idea of skipping transmission
  slots arises also in the recent work~\cite{sun2017update}, albeit in a slightly different context. We
  exhibit an example where a particular randomized scheme outperforms
  every deterministic scheme. However, our analysis is limited and does not
  completely clarify the role of randomization; for instance, it
  remains unclear for which distributions can randomized schemes
  strictly outperform deterministic ones.

In another direction, we consider the case where the transmission
channel is not error-free, but can erase each bit with a known
probability.  Furthermore, an ACK-NACK feedback indicating the success
of transmission is available.  Note that for the standard transmission
problem, the simple repeat-until-succeed scheme is optimal in this
setting.  Our analysis can be used to design the optimal source code
when we restrict our channel coding to this simple scheme.  However,
the optimality of the ensuing source-channel coding scheme remains
unclear.

Finally, we study the related problem of source coding for ensuring
minimum queuing delays.  This problem, introduced
in~\cite{Humblet78thesis}, is closely related to the minimum age
formulation of this paper.  Interestingly, our recipe for designing
update codes with minimum average age can be extended to this setting
as well.  However, here, too, our results are somewhat unsatisfactory:
Our approach only provides a solution to the real-relaxation of the
underlying integer-valued optimization problem and naive rounding-off
is far from optimal.  Nonetheless, we have included these extensions
in the current paper since they indicate the rich potential for our
proposed techniques and provide new formulations for future research.

The next section contains a formal description of our setting and a
formula for asymptotic average age of a code. Our main technical tool
is presented in Section~\ref{s:variational}, and we apply it to the
minimum average age code design problem in
Section~\ref{s:code_design}.  Numerical evaluations of our proposed
scheme for the family of Zipf distributions is presented in
Section~\ref{s:age_numerical}. Section~\ref{s:extensions} contains a
discussion on extensions to randomized schemes and erasure channel,
along with a treatment of source codes for minimum average waiting
time.  We provide all the proofs in the final section.

{\em Notation and some preliminaries.} Random variables are denoted by capital letters $X,Y$ etc., their realizations by small letters $x, y$ etc., \newer{and their range sets} by $\X, \Y$. 
The cardinality of the set $\X$ is denoted by $|\X|$. The set of all finite length binary sequences is denoted by $\{0,1\}^*$. 


The logarithm to the base 2 is denoted by $\log a$ and the logarithm to the base $e$ is denoted by $\ln a$. All the information theoretic measures considered in this paper -- such as Entropy, R\'enyi divergence, and  Kullback-Leibler divergence -- are defined with logarithm to the base 2.

Next, we recall the notions of Shannon lengths and Shannon codes, which will be used throughout. A source code is called {\em prefix-free} if no codeword is a prefix of another. 
\begin{defn}[Shannon lengths and Shannon codes for $P$]\label{d:shannon}
  For a pmf $P$ on an alphabet $\X$, the real-values
  $\ell(x)=-\log P(x), x \in \X$, are called the {\em Shannon lengths} for the pmf $P$. A prefix-free source code for $P$ with codeword lengths $\ell(x)=\ceil{-\log P(x)}, \ \forall x \in \X,$ is called a {\em Shannon code}\footnote{There can be different codes with codeword lengths required in our definition of a Shannon code. We simply refer to all of them as a Shannon code, since any of these can serve our purpose in this paper.} for the pmf $P$. 
\end{defn}

\section{Average age for memoryless update schemes}
\label{s:timely_updates} 

Consider a discrete-time system in which at every time instant $t$, a
transmitter observes a symbol $X_t$ generated from a finite alphabet
$\X$ with pmf $P$.  We assume that the sequence $\{X_t\}_{t=1}^\infty$
is {independent and identically distributed} (iid).  The transmitter
has a noiseless communication channel at its disposal over which it
can transmit one bit per unit time. A {\em memoryless update scheme}
consists of a prefix-free code, represented by its encoder $e:\X\to
\{0,1\}^*$, and a decoder which at each time instant $t$ declares a
time index $U(t)\leq t$ and an estimate $\hat{X}_{U(t)}$ for the
symbol $X_{U(t)}$ that was observed by the encoder at time $U(t)$. We
focus on error-free schemes and require $\hat{X}_{U(t)}$ to equal
$X_{U(t)}$ with probability $1$.

In a memoryless update scheme, once the encoder starts communicating a
symbol $x$, encoded as $e(x)$, it only picks up the next symbol once
all the bits in $e(x)$ have been transmitted successfully to the
receiver.  The time index $U(t)$ is updated to a new value only upon
receiving all the encoded bits for the current symbol. That is, if the
transmission of a symbol is completed at time $t-1$, \newer{the encoder will
start transmitting $e(X_t)$ in the next instant.} Moreover, if the
  final bit of $e(X_t)$ is received at time $t^\prime$, $U(t^\prime)$
  is updated to $t$. A typical sample path for $U(t)$ is given in
Figure~\ref{f:Age_sample_path}.
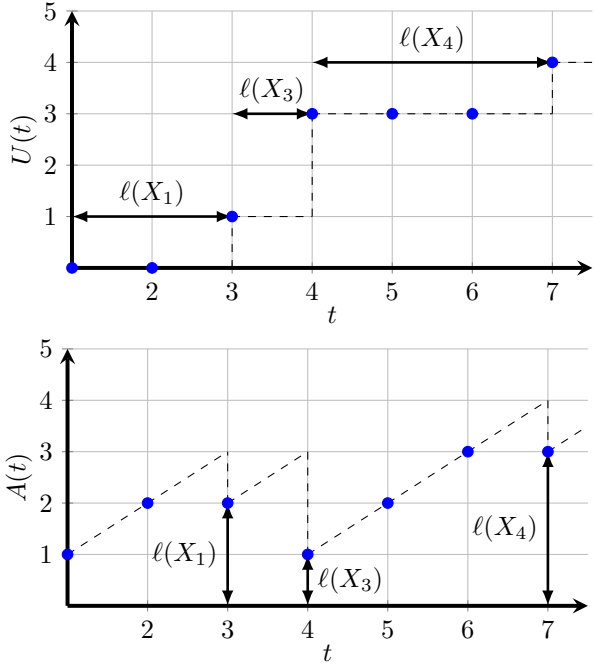
\begin{figure}
\centering
\begin{tikzpicture}[
soldot/.style={color=blue, only marks, mark=*}, 
dimarrow/.style={<->, >=latex, line width=1pt}
]

\begin{axis}[
width=8.5cm,height=5cm,
    axis line style={line width=1.5pt},
	xmajorgrids=true,
    ymajorgrids=true,
    grid style=solid,
	axis x line=middle, axis y line=middle,	
	y label style={at={(axis description cs:-0.05,.5)},rotate=90,anchor=south},
	x label style={at={(axis description cs:.5,-.15)},anchor=south},
    xmin=1, xmax=7.5,
    ymin=0, ymax=5,
    xtick={1,2,...,7},
    ytick={0,1,2,...,5},
	x label style={at={(axis description cs:.5,-.25)},anchor=south},
	xlabel={$t$},
    ylabel={$U(t)$}
]

\foreach \x in {1,2}
\addplot[soldot] coordinates{(\x,0)};
\foreach \x in {3}
\addplot[soldot] coordinates{(\x,1)};
\foreach \x in {4}
\addplot[soldot] coordinates{(\x,3)};
\foreach \x in {5}
\addplot[soldot] coordinates{(\x,3)};
\foreach \x in {6}
\addplot[soldot] coordinates{(\x,3)};
\foreach \x in {7}
\addplot[soldot] coordinates{(\x,4)};

\draw[dimarrow] (axis cs:1,1)  -- (axis cs:3,1)node[midway,above] {$\ell(X_1)$} ;
\draw[dashed] (axis cs:3,0)  -- (axis cs:3,1) ;

\draw[dimarrow] (axis cs:3,3)  -- (axis cs:4,3)node[midway,above] {$\ell(X_3)$} ;
\draw[dashed] (axis cs:3,1)  -- (axis cs:4,1) ;
\draw[dashed] (axis cs:4,1)  -- (axis cs:4,3) ;

\draw[dashed] (axis cs:4,3)  -- (axis cs:7,3) ;
\draw[dimarrow] (axis cs:4,4)  -- (axis cs:7,4)node[midway,above] {$\ell(X_4)$} ;
\draw[dashed] (axis cs:7,3)  -- (axis cs:7,4) ;
\draw[dashed] (axis cs:7,4)  -- (axis cs:7.5,4) ;
\end{axis}
\end{tikzpicture} 
\begin{tikzpicture}[
soldot/.style={color=blue, only marks, mark=*}, 
dimarrow/.style={<->, >=latex, line width=1pt}
]

\begin{axis}[
width=8.5cm,height=5cm,
    axis line style={line width=1.5pt},
	xmajorgrids=true,
    ymajorgrids=true,
    grid style=solid,
	axis x line=middle, axis y line=middle,	
	y label style={at={(axis description cs:-0.05,.5)},rotate=90,anchor=south},
	x label style={at={(axis description cs:.5,-.15)},anchor=south},
    xmin=1, xmax=7.5,
    ymin=0, ymax=5,
    xtick={1,2,...,7},
    ytick={0,1,2,...,5},
	x label style={at={(axis description cs:.5,-.25)},anchor=south},
	xlabel={$t$},
    ylabel={$A(t)$}
]

%
\foreach \x in {1,2}
\addplot[soldot] coordinates{(\x,\x)};
\foreach \x in {3}
\addplot[soldot] coordinates{(\x,2)};
\foreach \x in {4}
\addplot[soldot] coordinates{(\x,1)};
\foreach \x in {5}
\addplot[soldot] coordinates{(\x,2)};
\foreach \x in {6}
\addplot[soldot] coordinates{(\x,3)};
\foreach \x in {7}
\addplot[soldot] coordinates{(\x,3)};

\draw[dashed] (axis cs:1,1)  -- (axis cs:2,2) ;
\draw[dashed] (axis cs:2,2)  -- (axis cs:3,3) ;
\draw[dashed] (axis cs:3,2)  -- (axis cs:3,3) ;
\draw[dimarrow] (axis cs:3,0)  -- (axis cs:3,2) node[midway,left] {$\ell(X_1)$};
\draw[dashed] (axis cs:3,2)  -- (axis cs:4,3) ;
\draw[dashed] (axis cs:4,3)  -- (axis cs:4,1) ;
\draw[dimarrow] (axis cs:4,0)  -- (axis cs:4,1) node[midway,right] {$\ell(X_3)$};
\draw[dashed] (axis cs:4,1)  -- (axis cs:6,3) ;
\draw[dashed] (axis cs:6,3)  -- (axis cs:7,4) ;
\draw[dashed]  (axis cs:7,4) -- (axis cs:7,3);
\draw[dimarrow] (axis cs:7,3)  -- (axis cs:7,0) node[midway,left] {$\ell(X_4)$};
\draw[dashed]  (axis cs:7,3) -- (axis cs:7.5,3.5);
\end{axis}
\end{tikzpicture} 
\caption{A sample path of $U(t)$, $A(t)$ corresponding to Figure \ref{f:setup} starting with $U(1)=0$.}
\label{f:Age_sample_path}
\end{figure} 
The age $A(t)$ of the symbol available at the receiver at time $t$ is
given by
\begin{align}
A(t)= t-U(t).  \nonumber
\end{align}
A more general treatment can allow errors in estimates of $X_{U(t)}$ as
well as encoders with memory, but we limit ourselves to the
simple error-free and memoryless setting in this paper.

We are interested in designing prefix-free codes $e$ that minimize the
average age for the memoryless update scheme described above.
\begin{defn}
The {\it average age} for a prefix-free code $e$, denoted
$\bar{A}(e)$, is given by
\begin{align}
\bar{A}(e) = \limsup_{T\to \infty}\frac 1 T \sum_{t=1}^T {(t-U(t))}.
\nonumber
\end{align}
\end{defn}
\newer{We remark that  $\bar{A}(e)$ can be viewed as the average area under the curve of $A(t)$ (w.r.t. $t$). }
Note that $\bar{A}(e)$ is random variable, nevertheless we will prove
that this random variable is a constant almost surely.  
For any symbol $x \in \X$, we denote the length of the codeword $e(x)$ by $\ell(x)$. 
Let $X \in \X$ be a random symbol with pmf $P$ over the alphabet $\X$, 
then the length of the random codeword $e(X)$ is denoted by 
\[L= \ell(X).\] The result below
  uses a simple extension of the classical renewal reward theorem \newer{($cf.$ \cite{ross1996stochastic})} to
  provide a closed form expression for $\bar{A}(e)$ in terms of the
  first and the second moments of $L$.
\begin{thm}\label{t:average_age}
Consider a random variable $X$ with pmf $P$ on $\X$. For a prefix-free
code $e$, the average age $\bar{A}(e)$ is given by
\begin{align}
\bar{A}(e) = \E {L}+ \frac{\E{L^2}}{2\E {L}} - \frac 12 \quad
a.s. \quad.
\label{e:cost_theta}
\end{align}
\end{thm}
\noindent The proof is deferred to Section \ref{p:avg_age}.

Denoting by $\bar{A}^*$ the minimum average age over all prefix-free
codes $e$, as a corollary of the characterization above, we can obtain
the following bounds for $\bar{A}^*$.
\begin{cor}\label{c:age_bounds}
For any pmf $P$ over $\X$, the optimal average age $\bar{A}^*$ is
bounded as
\begin{align}
\frac 32 H(P) -\frac 12 \leq \bar{A}^* \leq \frac 32\log|\X| + 1.
\nonumber
\end{align} 
\end{cor}
The proof of lower bound simply uses Jensen's inequality $\E{L^2} \geq \E{L}^2$ and the fact
that $\E{L} \geq H(P)$ for a prefix free code; the upper bound is
obtained by using codewords of constant length $\lceil
\log|\X|\rceil$.

Note that the lengths $\ell(x)$ are required to be nonnegative
integers. However, for any set of real-valued lengths $\ell(x)\geq 0$,
we can obtain integer-valued lengths by using the rounded-off values
$\lceil \ell(x)\rceil$. Unlike the average length cost, the average
age cost function identified in \eqref{e:cost_theta} is not an
increasing function of the lengths. Nevertheless, by
\eqref{e:cost_theta}, the average age $\bar{A}(e)$ achieved when we
use the rounded-off values can be bounded as follows: Denoting
  $\bar{L}:=\lceil\ell(X)\rceil$, we have
\begin{align}
\E{\bar{L}} + \frac{\E{\bar{L}^2}}{2\E {\bar{L}}} - \frac 12 & \leq
\E{L+1} + \frac{\E{(L+1)^2}}{2\E {L}} -\frac{1}{2} \nonumber \\ \nonumber & \leq \E{L} +
\frac{\E{L^2}}{2\E {L}}+ \frac{2\E{L}}{2\E {L}}\\& \hspace{2cm}+ \frac{1}{2\E {L}}
+\frac 12 \nonumber \\ & \leq \E {L} + \frac{\E{L^2}}{2\E {L}} +2.
\label{e:ceil_loss}
\end{align}

Accordingly, in our treatment below we shall ignore the integer
constraints and allow nonnegative real-valued length assignments.

Returning now to the bound of Corollary~\ref{c:age_bounds}, the upper
and lower bounds differ only by a constant $1.5$ when $P$ is uniform. In view of the foregoing discussion, Shannon codes for
a uniform distribution attain the minimum average age up to a constant gap. The next result
gives an upper bound on average age for Shannon codes for an arbitrary
$P$ on $\X$.
\begin{lem}
Given a pmf $P$ on $\X$, a Shannon code $e$ for $P$ has average age $\overline{A}(e)$ at most $O(\log |\X|)$.
\end{lem}
\begin{proof}
\newer{Let $\ell(X)$ denote the lengths of Shannon code corresponding to $P$ (see Definition \ref{d:shannon}).} We establish the claim using the standard bound
  $H(P^\prime) \leq \log |\X|$ for an appropriately chosen pmf
  $P^\prime$ on $\X$. Specifically, for the tilting of $P$ given by
$P^\prime(x)\propto \ell(x)P(x)$, we get

 \begin{align*}
\log |\X| & \geq \sum_{x \in \X} \frac{ P(x)\ell(x)}{\E{\ell(X)}}
\log \frac{\E{\ell(X)}}{P(x)\ell(x)} 
\\ &= \sum_{x \in \X}
\frac{ P(x)\ell(x)(-\log P(x))}{\E{\ell(X)}} \\& \hspace{1cm} - \sum_{x \in \X}
\frac{P(x)\ell(x)}{\E{\ell(X)}}\log\frac{\ell(x)}{\E{\ell(X)}}
\\ &\geq \sum_{x \in \X} \frac{ P(x)\ell(x)(-\log P(x))}{\E{\ell(X)}}
\\ & \hspace{1cm} - \sum_{x \in \X: \ell(x) \geq
  \E{\ell(X)}}\frac{P(x)\ell(x)}{\E{\ell(X)}}\log\frac{\ell(x)}{\E{\ell(X)}}.
 \end{align*}
Using $-\log P(x) \geq \ell(x)-1$ and $\ln x \leq \frac{x^2-1}{2x}$
for $x\geq 1$, we obtain
 \begin{align*}
\log |\X| & \geq \frac{\E{\ell^2(X)}}{\E{\ell(X)}} - 1 \\&\hspace{0.5cm}- \frac
1{2\ln 2}\cdot \sum_{x \in \X: \ell(x) \geq \E{\ell(X)}} P(x)
\left(\frac{\ell^2(x)}{\E{\ell(X)}^2} - 1\right) 
\\
& \geq \frac{\E{\ell^2(X)}}{\E{\ell(X)}} - 1 \\&\hspace{1cm} - \frac
1{2\ln 2}\cdot \sum_{x \in \X: \ell(x) \geq \E{\ell(X)}} P(x)\cdot
\frac{\ell^2(x)}{\E{\ell(X)}^2}
\\ & \geq
\frac{\E{\ell^2(X)}}{\E{\ell(X)}} - 1 - \frac 1{2\ln 2}\cdot
\sum_{x \in \X}\frac{ P(x)\ell^2(x)}{\E{\ell(X)}^2}
\\& 
\geq \frac{\E{\ell^2(X)}}{\E{\ell(X)}} - 1 - \frac 1{2\ln 2}\cdot
\sum_{x \in \X}\frac{ P(x)\ell^2(x)}{\E{\ell(X)}}
 \\ &\geq
\left(1- \frac{1}{2\ln
  2}\right)\cdot\frac{\E{\ell^2(X)}}{\E{\ell(X)}} -1, 
 \end{align*}
where the second-last inequality follows from the fact that $\E{\ell^2(X)} \geq \E{\ell(X)}$, which in turn follows from the fact that $\ell(X) \geq 1$. The proof is
completed by rearranging the terms.
\end{proof}

It is of interest to examine if, in general, a Shannon code for $P$
itself has average age close to $\bar{A}^*$, as was the case for the
uniform distribution. In fact, it is not the case. Below we exhibit a
pmf $P$ where the average age of a Shannon code for $P$ is
$\Omega(\log |\X|)$, namely the previous bound is tight, and yet a
Shannon code for another distribution (when evaluated for $P$) has an
average age of only $O(\sqrt{\log |\X|})$.
\begin{example}\label{ex:1}
Consider $\X = \{0,...,2^n\}$ and a pmf $P$ on $\X$ given by \eq{
  P(x)=\begin{cases} 1-\frac 1n, \quad &x =0 \\ \frac 1{n2^n}, \quad
  &x \in \{1, \dots ,2^n\}.
\end{cases}
}Using \eqref{e:cost_theta}, the average age $\bar{A}(e_P)$ for a
Shannon code for $P$ can be seen to satisfy $\bar{A}(e_P) \approx (n +
2\log n)/2$. On the other hand, if we instead use a Shannon code
for the pmf $Q$ given by \eq{ Q(x)=
\begin{cases}
 \frac {1}{2^{\sqrt{n}}}, \quad &x =0 \\ \frac
       {1-2^{-\sqrt{n}}}{2^{n}}, \quad &x \in \{1, \dots ,2^n\},
\end{cases}
}we get $\E{L} \approx \sqrt{n}$ and $E{L^2} \approx 2n$, whereby
$\bar{A}(e_{Q}) \approx 2 \sqrt{n}$, just $O(\sqrt{\log
  |\X|})$.\qed
\end{example}
Thus, one needs to look beyond the standard { Shannon
  codes for $P$} to find codes with minimum average
age. Interestingly, we show that \newer{Shannon codes for a
  tilted version of $P$} attain the optimal asymptotic average age (up
to the constant loss of at most $2.5$ bits incurred by rounding-off
lengths to integers). In particular, for the example above, our
proposed optimal codes will have an average age of only $O(\sqrt{\log
  |\X|})$ in comparison to $\Omega(\log |\X|)$ of Shannon codes for
$P$.

\mnote{At HT: I  modified the comment above. Note that the optimal titling is not Q but something else.}
A key technical tool in design of our codes is a variational formula
that will allow us to linearize the cost function in
\eqref{e:cost_theta}, thereby rendering Shannon codes for a tilted
distribution optimal. We present this in the next section.
\section{A variational formula for  $p$-norm} \label{s:variational}
The expression for average age identified in
Theorem~\ref{t:average_age} involves the second moment of the random
codeword length $L$. This is in contrast to the traditional variable
length source coding problem where the goal is to minimize the average
codeword length $\E L$. For this standard cost, Shannon codes which
assign a codeword of length $\lceil-\log P(x)\rceil$ to the symbol $x$
come within $1$-bit of the optimal cost (see, for instance,
\cite{CovTho06}).  A variant of this standard problem was studied in
\cite{Campbell}, where the goal was to minimize the log-moment
generating function $\log \E {\exp(\lambda L)}$. A different approach
for solving this problem is given in \cite{SundaresanHanewal} where
the {\it Gibbs variational principle} is used to linearize the
nonlinear cost function $\log \E {\exp(\lambda L)}$. The next result
provides the necessary variational formula to extend the
aforementioned approach to another nonlinear function, namely
${\|L\|_p := }(\E {L^p})^{\frac 1p}$ for $p> 1$.

We believe that our result is of independent interest, and present it in a general 
  form that applies to general distributions (and not just the discrete random variables considered  in this paper). To state the general result, we recall a basic notation from probability theory. For two probability measures $P$ and $Q$ on the same probability space such that $Q$ is absolutely continuous with respect to $P$, denoted $Q\ll P$, denote by $\frac{d Q}{d P}$ the Radon-Nikodym derivative of $Q$ with respect to $P$. Note that $\frac{d Q}{dP}$, too, is a random variable measurable with respect to the underlying sigma-algebra.
  A reader not familiar with these notions can see a standard textbook on probability theory for definitions. For the discrete case, $Q\ll P$ corresponds to the condition\footnote{ \newer{${\tt supp}(P)$ denotes the support of distribution $P$ over an alphabet $\X$, $i.e.$,
${\tt supp}(P):=\{x \in \X: P(x)>0\}$.
      .}} ${\tt supp}(Q)\subset {\tt supp}(P)$ and $\frac{d Q}{d P}$ equals the ratio of the pmfs of the distributions $Q$ and $P$.

  Note that expectations are always taken with respect to the reference measure.  In particular, the expectations without any subscript in  Theorem~\ref{t:variational_formula} below and its proof denote the expectation with respect to $P$, which is the reference measure in this case. The expectation in Remark \ref{r:P(X=0)=0} denotes the expectation  with respect to $R$. 

\begin{thm}\label{t:variational_formula}
For a real-valued random variable $X$ with distribution $P$ and $p\geq 1$ such
that $\|X\|_p < \infty$, we have
\begin{equation}
\| X\|_p = \max_{Q\ll P} \E{ \left(\frac{dQ}{dP}\right)^{\frac 1
    {p^\prime}} |X|}, \nonumber
\end{equation}
where $p^\prime = p/(p-1)$ is the H\"older conjugate of $p$.
\end{thm}
\begin{proof}
For $Q\ll P$ and $0<\alpha\neq 1$, let $D_\alpha(P,Q)$ denote the
R\'enyi divergence of order $\alpha$ between distributions $Q$ and $P$
(see~\cite{Ren61}), defined by
\begin{equation*}
D_\alpha(P,Q) := \frac 1{\alpha-1}\log \E
{\left(\frac{dQ}{dP}\right)^\alpha}.
\end{equation*}
It is well-known that $D_\alpha(P,Q)\geq 0$ with equality if and only if
$P=Q$. Consider the probability measure $P_p\ll P$ defined by 
\begin{equation*}
\frac{d P_p}{dP} :=\frac 1{\|X\|_p^p}\cdot |X|^p.
\end{equation*}
Then, for $\alpha = 1/p^\prime$,

\begin{align*}
0&\leq D_\alpha(P_p,Q)
= \frac 1{\alpha-1} \log\E{\left(\frac{dQ}{dP}\right)^\alpha
  \left(\frac{dP_p}{dP}\right)^{1-\alpha}} \\ &=-p
\log\E{\left(\frac{dQ}{dP}\right)^\alpha |X|}+p \log \|X\|_p,
\end{align*}
where the previous equality holds since $p(1-\alpha)=1$.  Thus, for
every $Q\ll P$,
\begin{equation*}
 \E{\left(\frac{dQ}{dP}\right)^\alpha |X|}\leq \|X\|_p,
\end{equation*}
with equality if and only if $P_p=Q$.
\end{proof}

\begin{rem}\label{r:P(X=0)=0}
 The given definition of R\'enyi divergence restricts
    Theorem \ref{t:variational_formula} to the case $P(X=0)=0$. To
    remove this restriction, the following general definition of
    R\'enyi divergence with respect to a common measure can be
    used: For all $Q, P \ll R$, define
\[
    D_{\alpha}(P,Q):= \frac 1{\alpha-1}
    \log\E{\left(\frac{dQ}{dR}\right)^\alpha
      \left(\frac{dP}{dR}\right)^{1-\alpha}} .
\]
    The proof then follows by using the positivity of
    $D_{\alpha}(P_p,Q)$, then by proceeding in the same manner as
    the previous proof.
  \end{rem}
Returning to the problem at hand, \newer{we apply the variational formula above to the $L_2$ norm of a discrete random variable. We highlight this special case separately below.}

\begin{cor}\label{c:2variational_formula}
For a discrete random variable $X$ with a pmf $P$ such
that $ \|X\|_2 < \infty$, we have
\[
{
\|X\|_2=\max_{{\tt supp}(Q)\subset {\tt supp}(P)}\sum_{x\in \X}\sqrt{Q(x)P(x)}x},
\]
where ${\tt supp}(P)$ denotes the support-set of the distribution $P$. 
\end{cor}

\section{Prefix-free codes with minimum average age}
\label{s:code_design} 

We now present a recipe for designing prefix-free codes with minimum
average age.  By Theorem~\ref{t:average_age}, we seek prefix-free
codes that minimize the cost
\begin{equation}
\E L +\frac{\E {L^2}}{2\E L},
\label{e:cost_simple}
\end{equation}
where $L= \ell(X)$ for $X$ with pmf $P$.  Recall that a prefix-free code
with lengths $\{\ell(x) {\in \N}, x\in \X\}$ exists if and only if lengths
satisfy Kraft's inequality ($cf.$~\cite{CovTho06}), $i.e.$, if and only if
\begin{equation}
\sum_{x\in \X}2^{-\ell(x)}\leq 1.
\label{e:Kraft_cond}
\end{equation}
Following the discussion leading to \eqref{e:ceil_loss}, we relax the
{integral} constraints for $\ell(x)$ and search over all real-valued
$\ell(x)\geq 0$ satisfying \eqref{e:Kraft_cond}. Specifically, we solve the relaxed optimization problem
\begin{equation}\label{e:relaxed}
\min_{\ell \in \Lambda} \E{L}+\frac{\E{L^2}}{2\E{L}},
\end{equation}
where \[\Lambda =\big\{\ell \in \R^{|\X|}: \sum_{x \in \X} 2^{-\ell(x)} \leq1, ~\ell(x) \geq 0 ~\forall x \in \X\big\}.\] As noticed
in~\eqref{e:ceil_loss}, this can incur a loss of only a constant.  A
key challenge in minimizing \eqref{e:cost_simple} is that it is
nonlinear.  We linearize this cost as follows:
\begin{enumerate}
\item \newer{Note first the identity below, which is obtained by maximizing the expression on the right-side:}
\begin{align}\label{e:varZ}
&\E L +\frac{\E {L^2}}{2\E L} = \max_{z\geq 0} \left(1 - \frac
  {z^2}{2}\right) \E L + z \|L\|_2.
\end{align}

\item Then, Corollary~\ref{c:2variational_formula} yields
\[
  {\|L\|_2=\max_{Q\ll
      P}\sum_{x\in \X}\sqrt{Q(x)P(x)}\ell(x)},
  \]
  which further leads to 
\begin{align*}
  \lefteqn{\E L +\frac{\E  {L^2}}{2\E L}}
\\
  &= \max_{z\geq 0} \left(1 -\frac {z^2}{2}\right) \E L + z\max_{Q\ll
  P}\sum_{x\in \X}\sqrt{Q(x)P(x)}\ell(x)
\\
&=\max_{z\geq
    0}\max_{Q\ll P} \sum_{x\in \X}g_{z,Q,P}(x)\ell(x),
\end{align*}
where
\begin{align}
\label{e:g_def}
g_{z,Q,P}(x) &:= \left(1- \frac{z^2}2\right)P(x) +z\sqrt{Q(x)P(x)}.
\end{align}
\end{enumerate}
As remarked earlier, as the source distribution $P$ is discrete, the constraint $Q \ll P$ simplifies to ${\tt supp}(Q)\subset {\tt supp}(P)$.
Thus, our goal is to identify the minimizer $\ell^*$ that achieves
\begin{align}
\Delta^*(P) = \min_{\ell \in \Lambda}\max_{z\geq 0}\max_{Q\ll P}
\sum_{x\in \X}g_{z,Q,P}(x)\ell(x).
\label{e:minmax_cost}
\end{align}

The result below captures our main observation and facilitates the
computation of optimal lengths attaining the minmax cost $\Delta^*(P)$.  
\begin{thm}[Structure of optimal codes]\label{t:update_optimal}
The optimal minmax cost $\Delta^*(P)$ in \eqref{e:minmax_cost}
satisfies
\begin{align}
\Delta^*(P) &= \max_{z\geq 0}\max_{Q\ll P} \min_{\ell \in \Lambda}
\sum_{x\in \X}g_{z,Q,P}(x)\ell(x) \nonumber \\ &=\max_{\substack{z \geq 0,Q\ll P, \\ (z,Q) \in \mathcal{G}}} \sum_{x\in \X}g_{z,Q,P}(x)
\log\frac{\sum_{x^\prime\in \X}g_{z,Q,P}(x^\prime)}{g_{z,Q,P}(x)},
\label{e:maxmin_cost}
\end{align}
where \eq{ \mathcal{G}:=\{z \geq 0 ,Q \in \R^{|\X|}:g_{z,Q,P}(x)\geq0
  \quad \forall x \in \X\}.  }  Furthermore, if $(z^*, Q^*)$ is the
maximizer of the right-side of \eqref{e:maxmin_cost}, then the minmax
cost \eqref{e:minmax_cost} is achieved uniquely by {the Shannon
  lengths}\footnote{{Recall that Shannon lengths for the pmf $P$ on $\X$ are given by $\ell(x)=-\log P(x) $, $x \in \X$, and are not necessarily integers.}} for the pmf
$P^*$ on $\X$ given by
\begin{equation}\label{e:tilted_distribution}
P^*(x) = \frac {g_{z^*,Q^*,P}(x)}{\sum_{x^\prime\in
    \X}g_{z^*,Q^*,P}(x^\prime)}.
\end{equation}

\end{thm}
 Thus, our prescription for design of source codes is simple: Use
  a Shannon code for $P^*$ instead of $P$. To compute $P^*$, we need
  to solve the optimization problem in \eqref{e:maxmin_cost}. Note that is unclear a priori that the minimum average age for the problem in $\eqref{e:relaxed}$ would correspond to Shannon lengths for some pmf since our cost function is
  not monotonic in expected length, whereby the optimal solution may
  not satisfy Kraft's inequality with equality. Nonetheless, we show
  that the Shannon lengths $-\log P^*(x)$ are optimal for the relaxed problem given by \eqref{e:relaxed}.

\newer{We note that our formal result above only provides a structural result for the optimal solution. But we believe that this structural result leads to a recipe to design practical algorithms for finding the optimal solution; we describe this recipe below.}
  Specifically, note that the resulting optimization problem for finding $P^*$
is one of entropy maximization for which several heuristic recipes are
available. Furthermore, we note the following structural
simplification for the optimal solution which shows that if
$P(x)=P(y)$, then $P^*(x)=P^*(y)$ must hold as well; the proof is
relegated to the Appendix.  Thus, the dimension of the optimization
problem~\eqref{e:maxmin_cost} can be reduced from $|\X|+1$ to $M_P+1$,
where $M_P$ denotes the number of distinct elements in the probability
multiset $\{P(x): x \in \X\}$. Let $A_1 \cdots A_{M_P}$ denote the
partition of $\X$ such that
\[
P(x)=P(y) \quad \forall x,y \in A_i, \quad \forall i \in [M_P].
\]

\begin{lem}\label{c:dim_red}
Suppose that $Q^*$ is an optimal $Q$ for \eqref{e:maxmin_cost}. Then,
$Q^*$ must satisfy
\begin{align}\label{e:property}
Q^*(x)=Q^*(y) \quad \forall x,y \in A_i, \quad \forall i \in [M_P].
\end{align}
\end{lem}

In proving Lemma~\ref{c:dim_red}, we use the fact that the cost
function in \eqref{e:maxmin_cost} is concave in $Q$ for each fixed $z$
and is concave in $z$ for each fixed $Q$ (see Lemma
  \ref{c:concavity}). However, it may not be jointly concave in
$(z,Q)$. Nevertheless, we apply standard numerical packages to
optimize it in the next section to quantify the performance of our
proposed codes and compare it with Shannon codes for the original
distribution $P$.  
\section{Numerical results for Zipf distribution}\label{s:age_numerical}

We program all our optimization problems
  in \textit{AMPL} \cite{AMPL} and solve it using
  \textit{SNOPT}~\cite{SNOPT} and \textit{CONOPT}~\cite{CONOPT}
  solvers.
Specifically, for the pmfs $P$ we consider in this section, we solve the optimization problem given by  \eqref{e:maxmin_cost}  to find the corresponding optimal $(z^*, Q^*)$.  In order to check if we have indeed found the optimal $(z^*, Q^*)$, we once again use Theorem \ref{t:update_optimal}.
  In particular, it follows from Theorem \ref{t:update_optimal} that the necessary and sufficient condition for a particular $(z, Q)$ to be the optimal solution is that   the value of the maximization problem \eqref{e:maxmin_cost} at $(z, Q)$  equals
  \[\E{-\log P^{\prime}(X)}  +\frac{\E {(\log P^{\prime}(X))^2}}{2\E{-\log P^{\prime}(X)}},
  \]
  where
  \[
  P^{\prime}(X)=\frac {g_{z,Q,P}(x)}{\sum_{x^\prime\in
      \X}g_{z,Q,P}(x^\prime)};
  \]
in all our numerical evaluations, the solution found by the solver satisfies this condition, which establishes its optimality.

We now illustrate our recipe for construction of prefix-free codes
that yield minimum average age for memoryless update schemes when $P$
is a Zipf distribution. 
 Specifically, we illustrate our qualitative results using
the ${\tt Zipf}(s,N)$ distribution with alphabet $\X = \{1, \cdots, N\}$
and given by \eq{ P(i) = \frac{i ^{-s}}{\sum_{j=1}^N j^{-s}}, \quad
  1\leq i \leq N.  }

Heuristically, the  average age formula~\eqref{e:cost_theta} suggests that the differences between the performances of a code under average codeword length cost and the average age cost will be the most for ``peaky distribution,'' namely for distributions with heavy elements. The parameter $s$ of the Zipf distribution allows us to vary from a uniform distribution to a ``peaky distribution,'' making this family apt for our numerical study. Indeed, our numerical results confirm that our proposed scheme outperforms a Shannon code for $P$ when the parameter $s$ is high; see Figure~\ref{f:zipf_comparison}. When we
round-off real lengths to integers, the gains are subsided but still
exist. Further, when the parameter $s$ is close to $0$, { Shannon codes} for $P$ are close to
optimal.\newer{ With increase in $s$, the gain of our proposed schemes over Shannon codes starts becoming more prominent}. \newer{As an aside, Figure \ref{f:zipf_comparison} also provides an illustration of the non-monotonic nature of the average age function with respect to code lengths.}
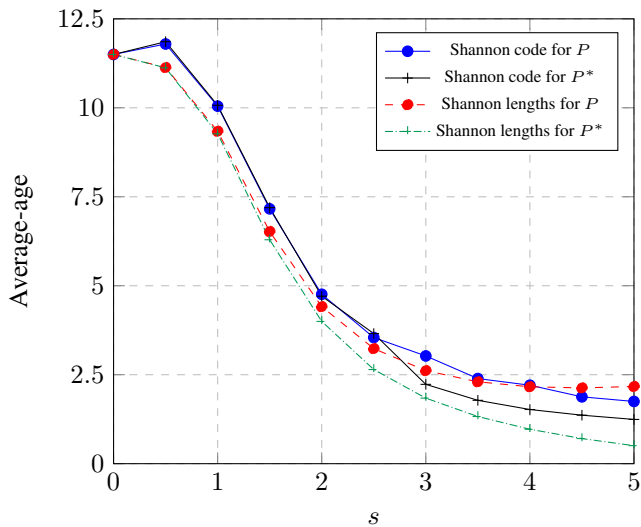
\begin{figure}[th]
\centering \begin{tikzpicture}[thick,scale=1, every node/.style={transform shape}]
\begin{axis}[ width=8.5cm,height=7.5cm,
    xlabel={$s$},
    ylabel={Average-age},
    xmin=0, xmax=5,
    ymin=0, ymax=12.5,
    xtick={0,1, 2,3, 4, 5},
    ytick={0,2.5,5,7.5,10,12.5},
    legend style={font=\fontsize{6.5}{8}\selectfont},
    legend pos=north east,
    xmajorgrids=true,
    ymajorgrids=true,
    grid style=dashed,
]

\addplot[color=blue , solid, mark=*, mark options={blue}]table [x=SV, y=PShannonAge, col sep=comma] {Figures/AOIZ_check.csv}; 
\addplot+[color=black ,solid, mark=+, mark options={black}] table [x=SV, y=PAvgAge, col sep=comma] {Figures/AOIZ_check.csv}; 
\addplot+[color=red ,dashed, mark=*, mark options={red}]table [x=SV, y=UBAvgAge, col sep=comma] {Figures/AOIZ_check.csv}; 
\addplot+[color=ForestGreen ,densely dashdotted, mark=+, mark options={ForestGreen}] table [x=SV, y=AvgAge, col sep=comma] {Figures/AOIZ_check.csv}; 
        \legend{Shannon code for $P$,Shannon code for $P^*$,
Shannon lengths for $P$,  Shannon lengths for $P^*$,}

\end{axis}
\end{tikzpicture}
\caption{ Comparison of proposed codes and Shannon codes for ${\tt
    Zipf}(s,256)$ with varying $s$.  The average age is computed using
  real-valued lengths as well as lengths rounded-off to integer
  values.} 

\label{f:zipf_comparison}
\end{figure}

The distribution $P^*$ we use to construct our codes seems to be a
flattened version of the original Zipf distribution; we illustrate the
two distributions for ${\tt Zipf}(1,8)$ in Figure~\ref{f:P_star}.
\begin{figure}[ht]
\centering \begin{tikzpicture}[thick,scale=1, every node/.style={transform shape}]
    \begin{axis}[
    width=8.5cm,height=6cm,
            ybar,            
            legend pos=north east,
            xtick={1,2,3,4,5,6,7,8},
            ymin=0,ymax=0.4,  
            ymajorgrids=true,
                grid style=dashed,        
        ]
        
             \addplot[
    black,
    fill=blue!60,
    postaction={
        pattern=north east lines
    }
] table [x=I, y=p, col sep=comma] {Figures/AOIZ_TD_S1_N8.csv};
             
             \addplot[
    black,
    fill=blue!20,
    postaction={
        pattern=north west lines
    }
]  table [x=I, y=tilteddistribution, col sep=comma] {Figures/AOIZ_TD_S1_N8.csv}; 
        \legend{$P$,$P^*$}
    \end{axis}
\end{tikzpicture}
\caption{The pmf for $P^*$ and $P$ for ${\tt Zipf}(1,8)$.}
\label{f:P_star} 
\end{figure}
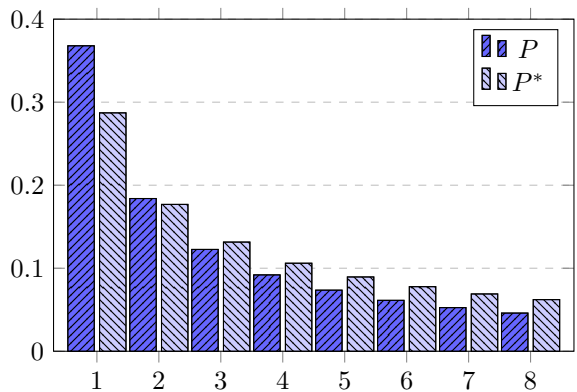
As we see in Figure~\ref{f:P_star}, $P^*$ and $P$ are very close in
this case.  Indeed, we illustrate in Figure~\ref{f:ent} that the
average length $\E{L}$ when Shannon lengths $-\log P(x)$ are used and
when $-\log P^*(x)$ are used are very close\footnote{The difference of these two average lengths (averaged
  w.r.t. $P$) is given by the Kullback-Leibler divergence $D(P\|P^*)$;
  see~\cite{CovTho06}. }.  In
Figure~\ref{f:ent}, we note the dependence of average age on the
entropy of the underlying distribution $P$. As expected, average age
increases as $H(P)$ increases.

\begin{figure}[h]
\centering \begin{tikzpicture}[thick,scale=1, every node/.style={transform shape}]
\begin{axis}[
    xlabel={$H(P)$},
    ylabel={},
    xmin=0, xmax=8,
    ymin=0, ymax=12,
    xtick={0,  2,  4, 6,8},
    ytick={0,2,4,6,8,10,12},
    legend style={font=\footnotesize},
    legend pos=north west,
    xmajorgrids=true,
    ymajorgrids=true,
    grid style=dashed,
]
 
\addplot[color=blue, mark=*] table [x=Entropy, y=CodeLength, col sep=comma] {Figures/AOIZ_check.csv};
\addplot[color=black ,solid, mark=*] table [x=Entropy, y=AvgAge, col sep=comma] {Figures/AOIZ_check.csv}; 
 \addplot +[color=gray ,densely dashdotted, mark= ]table [x=Entropy, y=Entropy, col sep=comma] {Figures/AOIZ_check.csv}; 
\legend{Average length (real lengths), Average age (real lengths), $Y=X$} 
\end{axis}
\end{tikzpicture}
\caption{ Average age and average length for our update codes as a
  function of $H(P)$ for ${\tt Zipf}(s,256)$ with $s$ varying from $0$
  to $5$ at step sizes of $0.5$.  }
\label{f:ent}
\end{figure}
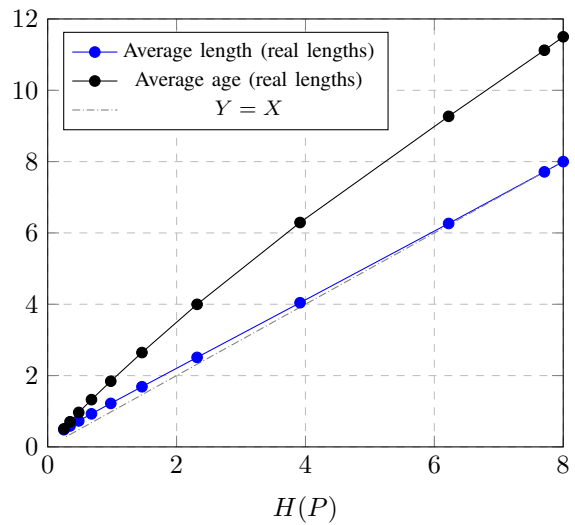

Thus, while Example~\ref{ex:1} illustrated high gains of the proposed
code over {Shannon codes} for $P$, for the specific case of Zipf
distributions the gains may not be large.  Characterizing this gain
for any given distribution is a direction for future research.

\section{Extensions}
\label{s:extensions}
\subsection{Randomization for Timely Updates}
\label{s:random_timely_updates}
We have restricted our treatment to deterministic {memoryless} update
schemes.  A natural extension to randomized memoryless schemes would
entail allowing the encoder to make a randomized decision to skip
transmission of a symbol even when the channel is free (we can
allocate a special symbol $\emptyset$ to signify no transmission to
the receiver).  Specifically, assume that we transmit the symbol
$\emptyset$ using a codeword of length $\ell(\emptyset)$ when we
choose not to transmit the observed symbol $x\in \X$. Denoting by
$\theta(x)$ the probability with which the encoder will transmit the
symbol $x$, the average age $\bar{A}(e,\theta)$ for the randomized scheme is given by
\begin{align}
\bar{A}(e,\theta)= \frac{\E {L(\theta)}}{\E{\theta(X)}}+
\frac{\E{L(\theta)^2}}{2\E {L(\theta)}} - \frac 12,
\label{e:randomized_avg_age}
\end{align}
where the random variable $L(\theta)$ is defined as follows:
\begin{align}\label{e:L(theta)} L(\theta) := \begin{cases}
      \ell(x),\quad w.p\quad P(x)\theta(x)\\ \ell(\emptyset),\quad
      w.p\quad 1-\E{\theta(X)}.
\end{cases}
\end{align}
Note that the expression in \eqref{e:randomized_avg_age} is a slight generalization of  Theorem~\ref{t:average_age} and is derived in Section~\ref{p:avg_age}.

\begin{example}
Consider $\X = \{1,...,64\}$ and the following pmf; \eq{P(x)=
 \begin{cases} 1/4, \quad & x \in
   \{1,\dots,3\},\\ 1/244, \quad & x \in \{4, \dots ,64\}. \end{cases}
} Since $H(P) = 3.483$, Corollary~\ref{c:age_bounds} yields that the
average age of the deterministic memoryless update scheme is bounded
below by $4.724$. Next, consider a randomized update scheme with
$\theta(x)=1$ for $x\in \{1,2,3\}$ and $0$ otherwise.
For this choice, the effective pmf $P_\theta$ is uniformly distributed
over the symbols $\{1,2,3\}\cup\{\phi\}$. Thus, the optimal length
assignment for this case assigns $\ell(x)=2$ to all the symbols and
the average age equals $3.17$, which is less than the lower bound of
$4.724$ for the deterministic scheme.
\end{example}
The idea of skipping available transmission opportunities, i.e., not
transmitting even when the channel is free, to
  minimize average age appears in the recent work~\cite{sun2017update}
  as well, albeit in a slightly different setting.  Heuristically,
the randomization scheme above operates as we expect -- it ignores the
rare symbols which will require longer codeword lengths. In practice,
however, these rare symbols might be the ones we are interested
in. But keep in mind that our prescribed solution only promises to
minimize the average age and does not pay heed to any other
consideration. Furthermore, for a given randomization vector
  $\theta$, we can establish a result similar to
  Theorem~\ref{t:update_optimal}. This will lead to the design of almost
  optimal source codes for a given randomization vector $\theta$.
However, the joint optimality over the class of randomized schemes and
source coding schemes is still unclear.

 In a more comprehensive
treatment, one can study the design of update codes with other
constraints imposed. We foresee the use of
Corollary~\ref{c:2variational_formula} in these more general settings as
well.  
In another direction, we can consider the extension of our results to
the case when the transmission channel is an erasure channel with
probability of erasure $\ep$. If we assume the availability of perfect
feedback, a natural model for the link or higher layer in a network,
and restrict to simple repetition schemes where the transmitter keeps
on transmitting the coded symbol until it is received, our formula for
average age extends with (roughly) an additional multiplicative factor
of $1/(1-\ep)$.  Formally the average age over an erasure channel with
$\ep$ probability of erasure; a source code $e$, along with a
randomization vector $\theta$ and a repetition channel-coding scheme
yields the following average age
\begin{align*}
\bar{A}_\ep(e,\theta) = \frac 1 {1-\ep}\cdot \bar{A}(e,\theta) +
\frac{\ep}{2(1-\ep)}.
\end{align*} However, the optimality of repetition scheme is unclear, and 
the general problem constitutes a new formulation in joint-source
channel coding which is of interest for future research.

\subsection{Source Coding for Minimum Queuing Delay}
\label{S:Min_Del}
Next, we \newer{point out} a use case for Corollary~\ref{c:2variational_formula}
in a minimum queuing delay problem introduced
in~\cite{Humblet78thesis}.  The setting is closely related to our
minimum average age update formulation with two differences: First,
the arrival process of source symbols is a Poisson process of rate
$\lambda$; and second, the encoder is not allowed to skip source
symbols.  Instead, each symbol is encoded and scheduled for
transmission in a first-come-first-serve (FCFS) queue.  Our goal is to
design a source code that minimizes the average queuing delay
encountered by the source sequence.  Formally, the symbols
$\{X_n\}_{n=1}^\infty$ are generated iid from a finite alphabet $\X$,
using a common pmf $P$. Every incoming symbol $x$ is encoded as $e(x)$
using a prefix-free code specified by the encoder mapping $e:\X
\to \{0,1\}^*$, and the bit string $e(x)$ is placed in a queue. The
queue schedules bits for transmission using a FCFS policy. Each bit in
the queue is transmitted over a noiseless communication channel.
Denote by $A_n$ the time of successful arrival of the $n$th symbol. Also, denote by $D_n$ the time instant of successful reception of the $n$th
symbol $X_n$.  {That is, }$D_n$ is the instant at which the last bit
of $e(X_n)$ is received\footnote{\newer{Note both $A_n$ and $D_n$ may not be integer valued, unlike the age setup.}}. The delay for the $n$th symbol is given by
$D_n - A_n$; see Figure~\ref{f:min_avg_age} for an illustration.

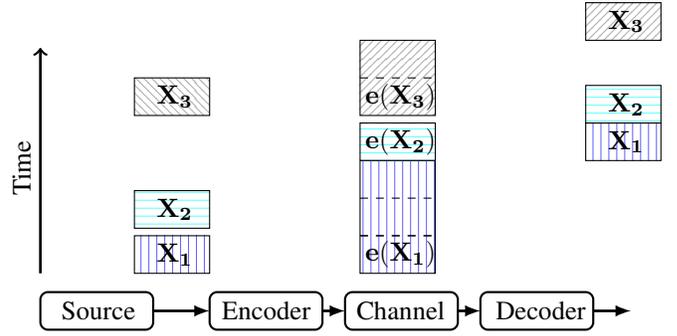
\begin{figure}[h]
\centering \begin{tikzpicture}
[scale=1, transform shape,
    pre/.style={=stealth',semithick},
    post/.style={->,shorten >=1pt,>=stealth',semithick},
dimarrow/.style={->, >=latex, line width=1pt}
    ]
\draw[<-,line width=1pt] (0.25,3) to (0.25,0);
\node[align=center, rotate=90] at (0.0,1.4) {Time}; 
\draw[thick, rounded corners=1.mm] (0.25, -0.25) rectangle (1.75, -.75);  
\node[align=center] at (1.02,-0.5) {Source};  
\draw[dimarrow] (1.75,-0.5) to (2.5,-0.5);
\draw[thick, rounded corners=1.mm] (2.5, -0.25) rectangle (4, -.75);
\node[align=center] at (3.25,-0.5) {Encoder};
\draw[dimarrow] (4,-0.5) to (4.3,-0.5);
\draw[thick, rounded corners=1.mm] (4.3,-0.25) rectangle (5.8, -0.75);
\node[align=center] at (5.1,-0.5) {Channel };
\draw[dimarrow] (5.8,-0.5) to (6.1,-0.5);
\draw[thick, rounded corners=1.mm] (6.1, -0.250) rectangle (7.6, -0.75);
\node[align=center] at (6.9,-0.5) {Decoder};
\draw[dimarrow] (7.6,-0.5) to (8.1,-0.5);
 
\draw[pattern=vertical lines, pattern color=blue!60] (1.5,0) rectangle (2.5,0.5); 
\node[align=center] at (2.1,0.25) {$\bold{X_1}$ };

\draw[pattern=horizontal lines, pattern color=cyan!60] (1.5,0.6) rectangle (2.5,1.1); 
\node[align=center] at (2.1,0.85) {$\bold{X_2}$ };


\draw[pattern=north west lines, pattern color=gray!60] (1.5,2.1) rectangle (2.5,2.6); 
\node[align=center] at (2.1,2.35) {$\bold{X_3}$ };


\draw[pattern=vertical lines, pattern color=blue!60] (4.5,0) rectangle (5.5,1.5); 
\node[align=center] at (5.1,0.25) {$\bold{e(X_1)}$ };
\draw [dashed] (4.5,0.5) to (5.5,0.5);
\draw [dashed] (4.5,1) to (5.5,1);

\draw[pattern=horizontal lines, pattern color=cyan!60] (4.5,1.5) rectangle (5.5,2); 
\node[align=center] at (5.1,1.75) {$\bold{e(X_2)}$ };

\draw[pattern=north east lines, pattern color=gray!60] (4.5,2.1) rectangle (5.5,3.1); 
\node[align=center] at (5.1,2.35) {$\bold{e(X_3)}$ };


\draw [dashed] (4.5,0.5) to (5.5,0.5);
\draw [dashed] (4.5,2) to (5.5,2);
\draw [dashed] (4.5,2.6) to (5.5,2.6);


\draw[pattern= vertical lines, pattern color=blue!60] (7.5,1.5) rectangle (8.5,2); 
\node[align=center] at (8.1,1.75) {$\bold{X_1}$ };

\draw[pattern=horizontal lines, pattern color=cyan!60] (7.5,2) rectangle (8.5,2.5); 
\node[align=center] at (8.1,2.25) {$\bold{X_2}$ };

\draw[pattern=north east lines, pattern color=gray!60] (7.5,3.1) rectangle (8.5,3.6); 
\node[align=center] at (8.1,3.35) {$\bold{X_3}$ };
\end{tikzpicture}
\caption{ Figure describes a typical sample-path for transmission of
  encoded symbols over a FCFS queuing system. Symbol $X_1$ arrives at
  some time instant $1$, it is encoded and transmitted over the
  channel. \newer{Recall that unlike the slotted setup of Figure~\ref{f:setup},
    the setup here is that of continuous time with Poisson arrivals.} 
  It is decoded at time instant $4$.  Symbol $X_2$ arrives
  in between time instants $2$ and $3$, and is placed in the queue, as
  the channel is busy transmitting $X_1$.  As soon as the channel
  becomes free at time instant $4$, an encoded version of $X_2$ is
  transmitted over it.
  Symbol $X_3$ arrives when the channel is free
  and is transmitted immediately.}
\label{f:min_avg_age}
\end{figure}

Thus, if $\ell(x)$ is the length of the encoded symbol $e(x)$ in bits,
then the number of channel uses to transmit this symbol is $\ell(x)$,
whereby the service time of the $n^{th}$ arriving symbol is given by
$S_n = \ell(X_n)$. Since $\{X_n\}_{n=1}^\infty$ is iid and the encoder
mapping $e$ is fixed, the sequence $(S_n)_{n \in \N}$, too, is iid
with common mean $\E{L}$. Therefore, the resulting queue is an M/G/1
queuing system with Poisson arrivals of rate $\lambda$ and iid service
times $(S_n)_{n \in \N}$.  Note that this queue will be stable only if
$\lambda\E{S_n} = \lambda\E{L}< 1$.

We are interested in designing prefix-free  codes $e$ that
minimize the average waiting time defined as follows:
\begin{defn}
The {\it average {waiting time}} $D(e)$ of a source code $e$ is given
by
\begin{align*}
D(e) &:= \limsup_{N\to \infty}\frac 1 N \sum_{n=1}^N \E{D_n-A_n},
\end{align*}
where the expectation is over source symbol realizations
$\{X_n\}_{n=1}^\infty$ and arrival instants $\{A_n\}_{n \in \N}$.
\end{defn}
We seek prefix-free codes $e$ with the least possible average
waiting time $D(e)$.  In fact, a closed-form expression for $D(e)$ was
obtained in~\cite{Humblet78thesis}.  For clarity of exposition, we
denote the load for the queuing system above for a fixed $\lambda$ by
$\rho(L) {:=} \lambda\E{L}$.  Since $\rho(L) < 1$ for the queue to be
stable, the average codeword length $\E L$ must be strictly less than
a threshold $L_{\tt th} {:=} \frac{\E{L}}{\rho(L)} =
\frac{1}{\lambda}$ for the queue to be stable.
\begin{thm}[\cite{Humblet78thesis}]\label{t:average_delay}
Consider a random variable $X$ with pmf $P$ and a source code $e$
which assigns a bit sequence of length $\ell(x)$ to $x\in \X$.  Let
$L$ denote the random variable $\ell(X)$.  Then, the average waiting
time $D(e)$ for $e$ is given by
\begin{align}
D(e) &= \begin{cases} \frac{ \E{L^2}}{2(L_{\tt th} - \E{L})} +
  \E{L} , &\E{L} < L_{\tt th},\\ \infty, &\E{L} \geq L_{\tt
    th}.
\end{cases}
%
\label{e:delay_formula}
\end{align}
\end{thm}
Thus, the problem of designing source codes with minimum average
waiting time reduces to that of designing a prefix-free code that
minimizes the cost in \eqref{e:delay_formula}. This problem was first considered in \cite{Humblet78thesis}.  In fact, it was noted
in~\cite[Chapter 1, Section 3]{Humblet78thesis} that codes which
minimize the first moment are robust for \eqref{e:delay_formula}.
We will justify this empirical observation in Corollary
  \ref{c:KL-bound}.  However, optimal codes can differ from Shannon
codes for $P$.  Indeed, an algorithm for finding the optimal length
assignments $\ell(x)$, $x\in \X$, for a prefix-free code that
minimizes $\bar{D}(e)$ was presented in \cite{larmore1989minimum} and
the optimal code can be seen to outperform Shannon codes for $P$.
While this algorithm has complexity that is polynomial in the alphabet
size, it is computationally expensive for large alphabet sizes -- the
case of interest for our problem.

Interestingly, the cost function in \eqref{e:delay_formula} resembles
closely the expression we obtained for asymptotic average age and our
recipe used to design minimum average age codes can be applied to
design minimum average delay codes as well. The underlying
optimization problem can be solved numerically rather quickly, much
faster than the optimization in \cite{larmore1989minimum}.  However,
as before, our procedure can only handle the real-relaxation of the
underlying optimization problem, and unlike the previous case, naive
rounding-off to integer lengths yields a sub-optimal solution when
$(1-\rho(L))$ is small. Nonetheless, the minimum average waiting time
computed using our recipe serves as an easily computable lower bound
for the optimal $D(e)$. In fact, we observe in our numerical
simulations that the resulting lower bound is rather close to the
optimal cost obtained using~\cite{larmore1989minimum}.

Now, we describe the modification of our recipe to design codes with
$\E{L} < L_{\tt th}$ that minimize the cost
\begin{align}
\norm{L}_1 +\frac{\norm{L}_2^2 }{2(L_{\tt th}-\norm{L}_1)},
\label{e:qcost_simple}
\end{align}
where $L= \ell(X)$ for $X$ with pmf $P$. As before, we first obtain a
variational form of \eqref{e:qcost_simple} which entails a linear
function of lengths. Specifically, we have the following steps.
\begin{enumerate}
\item First, we obtain a polynomial form from the rational function:
  \eq{
&\frac{\norm{L}_2^2}{2(L_{\tt th} - \norm{L}_1)} = \max_{z \geq 0} z
    \norm{L}_2 - \frac{z^2}{2} (L_{\tt th} - \norm{L}_1).
}

\item 
Then, Corollary~\ref{c:2variational_formula} yields that the cost
in~\eqref{e:qcost_simple} equals
\begin{align*}
&\max_{z\geq 0}\max_{Q\ll P} \sum_{x\in \X}g_{z,Q,P}(x)\ell(x) -
  \frac{z^2}{2} L_{\tt th}
\end{align*}
where the $g_{z,Q,P}(x)$ is defined as
\begin{align*}
g_{z,Q,P}(x) &:= \left(1+ \frac{z^2}2\right)P(x)+z\sqrt{Q(x)P(x)}.
\end{align*}
\end{enumerate}
Thus, our goal reduces to identifying the minimizer $\ell^* \in
\Lambda$ that achieves
\begin{align}
 \Delta^*(P) = \min_{\substack{\ell\in \Lambda,\\ \E{L}<L_{\tt th}}
 }\max_{z\geq 0}\max_{Q\ll P} \sum_{x\in \X}g_{z,Q,P}(x)\ell(x) -
 \frac{z^2}{2} L_{\tt th}.\label{q:minmax_cost}
\end{align}
The result below is the counterpart of Theorem~\ref{t:update_optimal}
for minimum delay source codes and is proved in Section~\ref{p:3}.
\begin{thm}\label{t:qmaintheorem} 
Under the condition
\begin{equation}\label{eq:a_qmaintheorem}
H(X) + \log(1 + 1/\sqrt{2})< L_{\tt th},
\end{equation}
the optimal minmax cost $\Delta^*(P)$ in \eqref{q:minmax_cost}
satisfies
\begin{align}
\Delta^*(P) &= \max_{z\geq 0}\max_{Q\ll P} \min_{\substack{\ell\in
    \Lambda,\\ \E{L}<L_{\tt th}} } \sum_{x\in \X}g_{z,Q,P}(x)\ell(x)-
\frac{z^2}{2} L_{\tt th} \nonumber \\ &=\max_{z\geq 0}\max_{Q\ll P}
\sum_{x\in \X}g_{z,Q,P}(x)\log \frac{\sum_{x^\prime\in
    \X}g_{z,Q,P}(x^\prime)}{g_{z,Q,P}(x)} \nonumber \\ & \hspace{5cm} - \frac{z^2}{2} L_{\tt th}.
\label{q:maxmin_cost}
\end{align}
Furthermore, if $(z^*, Q^*)$ is the maximizer of the right-side of
\eqref{q:maxmin_cost}, then the minmax cost \eqref{q:minmax_cost} is
achieved uniquely by Shannon lengths for pmf
$P^*$ on $\X$ given by
\[
P^*(x) = \frac {g_{z^,Q^*,P}(x)}{\sum_{x^\prime\in
    \X}g_{z^*,Q^*,P}(x^\prime)}.
\]
\end{thm}
We remark that  $\eqref{e:delay_formula}$ implies that $H(X) < L_{\tt th}$ is essential for the existence of a prefix free source coding scheme with finite average delay. Thus, the condition $H(X)+ \log(1  + 1/\sqrt{2}) < L_{\tt th}$ is a mild one.

Thus, as before, the optimal codeword lengths for the relaxed problem (allowing real-valued lengths) correspond, once again, to Shannon lengths for a titled distribution
$P^*$.  As remarked earlier, the performance of the optimal source
code is known to be not too far from the Shannon code for $P$. This
observation can be justified by the following simple corollary of
Theorem~\ref{t:qmaintheorem}.
\begin{cor} \label{c:KL-bound}
The KL-Divergence between $P$, $P^*$ is bounded as \eq{D(P||P^*) \leq
  \log\left(1+\frac{1}{\sqrt{2}}\right) .}
  \label{q:corr}
\end{cor}
\begin{proof}
The proof follows from \eqref{e:q_corr}, which is in turn derived in
the proof of theorem \ref{t:qmaintheorem} in section \ref{p:3}.
\end{proof}
Thus, the average length for Shannon codes and our codes do not differ
by more than $\log(1+1/{\sqrt{2}})$ ($cf.$~\cite{CovTho06}).
Indeed, we note in
  Figures~\ref{f:MinDelay_comparisona},~\ref{f:MinDelay_comparisonb}
  via numerical simulations that the optimal cost
  in~\eqref{q:maxmin_cost} is very close to the performance of optimal
  codes designed using~\cite{larmore1989minimum}. \newer{This suggests that 
    possibly there is 
an appropriate rounding-off procedure for real-valued lengths that can
yield integer lengths with close to optimal performance}; devising such a
rounding-off procedure is an interesting research direction for the
future. We close this section by noting that analogous versions of Lemma~\ref{c:dim_red}
and Lemma~\ref{c:concavity} in the Appendix can be obtained for optimization problem \eqref{q:maxmin_cost}.

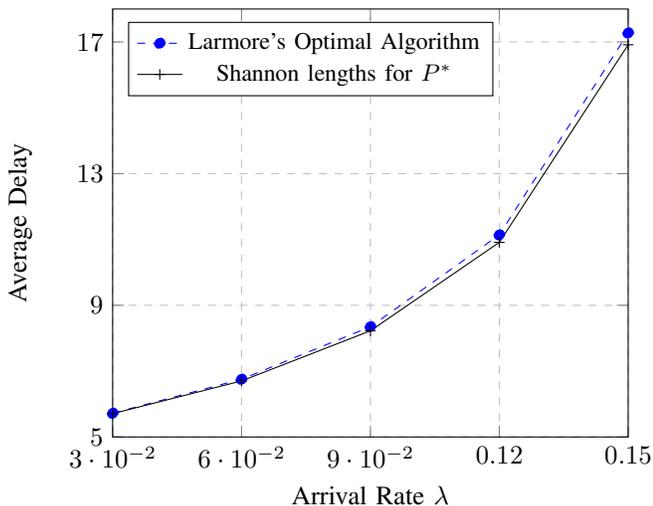
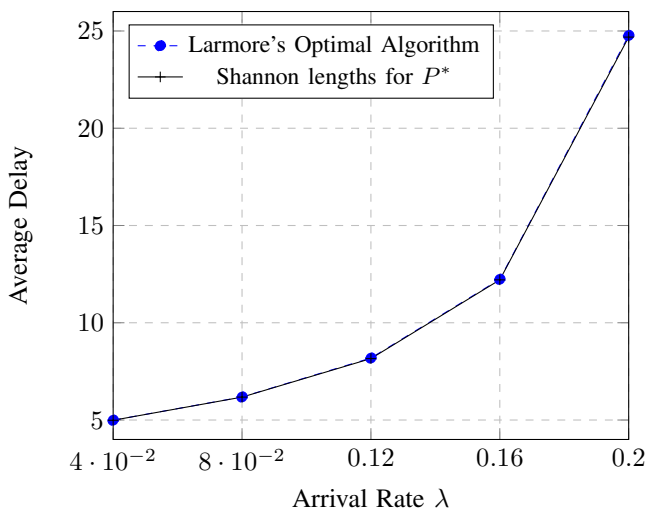
\begin{figure}[ht]
\centering
\begin{subfigure}[b]{\columnwidth}
\centering \begin{tikzpicture}[thick,scale=1, every node/.style={transform shape}]
\begin{axis}[
    xlabel={Arrival Rate $\lambda$},
    ylabel={Average Delay},
    xmin=0.03, xmax=.15,
    ymin=5, ymax=18,
    xtick={ .03,  .06, .09,.12, .15},
    ytick={5,9,13,17},
    legend style={font=\fontsize{9}{10}\selectfont},
    legend pos=north west,
    xmajorgrids=true,
    ymajorgrids=true,
    grid style=dashed,
]

\addplot+[color=blue, dashed, mark=*, mark options={blue}]table [x=ArrivalRate, y=OptimalDelay(LarmoresAlgo), col sep=comma] {Figures/Comprsn_Shan_MinDel_Codes.csv}; 
\addplot+[color=black, solid, mark=+] table [x=ArrivalRate, y=maxminLowerBound, col sep=comma] {Figures/Comprsn_Shan_MinDel_Codes.csv}; 
        \legend{
Larmore's Optimal Algorithm,  Shannon lengths for $P^*$,}
\end{axis}
\end{tikzpicture}
\caption{Comparison of proposed codes with
  Larmore's Algorithms \cite{larmore1989minimum} for the distribution
  $P(1)=0.5$, and $P(i)=\frac{0.5}{ 255} \quad \forall i \in
  \{2,\cdots 256\}$.}
\label{f:MinDelay_comparisona}
\end{subfigure}
\hspace{0.1cm}

\begin{subfigure}[b]{\columnwidth}
\centering \begin{tikzpicture}[thick,scale=1, every node/.style={transform shape}]
\begin{axis}[
    xlabel={Arrival Rate $\lambda$},
    ylabel={Average Delay},
    xmin=0.04, xmax=.20,
    ymin=4, ymax=26,
    xtick={ .04,  .08, .12,.16, .20},
    ytick={5,10,15,20,25},
    legend style={font=\fontsize{9}{10}\selectfont},
    legend pos=north west,
    xmajorgrids=true,
    ymajorgrids=true,
    grid style=dashed,
]

\addplot+[color=blue ,dashed,  mark=*, mark options={blue}]table [x=ArrivalRate, y=OptimalDelay(LarmoresAlgo), col sep=comma] {Figures/Comprsn_Shan_MinDel_Codes_1.csv}; 
\addplot+[color=black, solid, mark=+] table [x=ArrivalRate, y=maxminLowerBound, col sep=comma] {Figures/Comprsn_Shan_MinDel_Codes_1.csv}; 
        \legend{ 
Larmore's Optimal Algorithm,  Shannon lengths for $P^*$,}
                                                                                                                                                                                                                                                                                                                                                                                                                                                                                                                                                                                                                                                                                                                                                                                                                                                                                                                                                                                                                                                                                                                                                                                                                                                                                                                                                                                                                                                                                                                                                                                                                                                                                                                                                                                                                                                                                                                                                                                                                                                                                                                                                                                                                                                                                                                                                                                                                                            \end{axis}
\end{tikzpicture}
\caption{Comparison of proposed codes with 
  Larmore's Algorithms \cite{larmore1989minimum} for the distribution
  $P(1)=0.6$, and $P(i)=\frac{0.4}{ 255} \quad \forall i \in
  \{2,\cdots 256\}$.}
\label{f:MinDelay_comparisonb}
\end{subfigure}
\caption{Comparison of proposed codes with  Larmore's Algorithms}
\end{figure}
\section{Proofs}
\subsection{Proof of Theorem~\ref{t:average_age}}\label{p:avg_age}
We establish the expression for average age given in
\eqref{e:randomized_avg_age} for the more general class of randomized
schemes; Theorem~\ref{t:average_age} will follow upon setting
$\theta(x)=1$, for all $x \in \mathcal{X}$. Recall that the symbol $\emptyset$ is available only in the extended model in Section~\ref{s:extensions}, and not in the original model discussed in rest of the paper.
Note that the formula for
average age given in Theorem~\ref{t:average_age} is similar in form to
the expressions for average age derived in other settings;
see~\cite{KaulYatesGruteser12} for an example.

We will first set up some notation.
Let $S_0 := 0$ and \eq{ &S_k := \inf\{t > S_{k-1} :
  U(t)>U(t-1)\},~k \in \N.  }
\newer{Namely, $S_k$ is the time at which the decoder updates
its estimate for the symbol for the $k$th time.
Recall that $U(t)$ is incremented only
on successful reception at the receiver and is strictly increasing in $t$.
For brevity, we introduce the
notation $Y_k := S_k - S_{k-1}$
for the time between the $(k-1)$th and the $k$th information update at the decoder.
Further, denote by $Z_k:=S_{k}-U(S_k)$ the age at time $S_k$, which is simply the time taken for
  the successful reception of the symbol\footnote{This must be a symbol in $\X$ and not $\emptyset$ by the definition of $S_k$.} $x\in \X$ transmitted at time $U(S_k)$.}
Also, denote by $R_k$ \newer{the sum of instantaneous age  between $S_{k-1}$ and $S_k$} (the $k$th reward), namely
\eq{
R_k:= \sum_{t=S_{k-1}+1}^{S_{k}}(t-U(t)).
}
\mnote{"For Brevity," This  sentence is confusing too me. I preferred the preivious iteration of this sentence.}
\tnote{Is it okay now?}

Heuristically, our proof can be understood as follows. We note that the asymptotic average age is
roughly
\[
\frac{\sum_{k=1}^\infty R_k}{\lim_{k\to \infty} S_k}.
\]
It is easy to see that $\{Y_k\}_{k=1}^{\infty}$ is an iid sequence. Thus, if
$\{R_k\}_{k=1}^\infty$, too, was an iid sequence, we would obtain the
asymptotic average age to be $\E{R_1}/{\E{Y_1}}$ by the standard
Renewal Reward Theorem~\cite{ross1996stochastic}.  Unfortunately, this
is not the case.  But it turns out that the dependence in sequence
$\{R_k\}$ is only between consecutive terms.  Therefore, we can obtain
the same conclusion as above by dividing the sum $\sum_{k=1}^\infty
R_k$ into the sum of odd terms and even terms, each of which is in
turn a sum of iid random variables.

We will now proceed to prove that dependence in $R_k$ is between consecutive terms.  
Since $U(t)$ remains $U(S_{k-1})$ for all $t<S_k$, we get for $k\geq
1$ that
\begin{align}\label{e:Rk}
\nonumber
R_k &= \frac{(S_{k}-S_{k-1} -1)(S_{k}-S_{k-1})}{2} \\ \nonumber&\hspace{1cm} +(S_{k}-S_{k-1} 
  -1)\cdot(S_{k-1}-U(S_{k-1})) \\& \hspace{4.5cm} +S_{k}-U(S_{k}) \nonumber
\\
& =\frac 12 Y_k^2+Y_{k}\left(Z_{k-1} -  \frac 12 \right)+Z_k - Z_{k-1},
\end{align}
with $Z_0$ set to $0$.

Note that since the source sequence $\{X_n\}$ is iid and the randomization $\theta$ is stationary, the sequences $Y_k$ and $Z_k$
are iid, too. 
Therefore, the $(R_{2n})_{n \in \N}$ and $(R_{2n+1})_{n \in \N}$ are both\footnote{The initial term $R_1$ has a different distribution since $Z_0=0$.} 
iid sequences with $\E{R_{2n}} = \E{R_{2n+1}} = \E{R_2}$ for all $n$.   

Using this observation, we can obtain the following expression for the average age:
\begin{equation}
\bar{A}(e,\theta)=\frac{\E{R_2}}{\E{Y_1}}.
\label{eq:avgage_RY}
\end{equation}
Before we prove \eqref{eq:avgage_RY}, which is the main ingredient of
our proof, we evaluate the expression on the right-side. 

For $\E{Y_1}$, note that $Y_1$ gets incremented by $\ell(\emptyset)$ each
time $\emptyset$ is sent, and gets incremented finally by $\ell(x)$
once a symbol $x\in\X$ is sent. Thus, $Y_1$ takes the value
$\ell(x)+r\ell(\emptyset)$ with probability
$(1-\E{\theta(X)})^r\theta(x)P(x)$.
Denoting $\N_0=\N\cup \{0\}$, we get
\eq{
\E{Y_1}&= \sum_{x \in \X}\sum_{r \in \N_0} (\ell(x)+r\ell(\phi))P(x)\theta(x) (1-\E{\theta(X)})^r\\
&=\sum_{x \in \X} \sum_{r \in \N_0} \ell(x)P(x)\theta(x)(1-\E{\theta(X)})^r  \\&\hspace{1.3cm}+  
\sum_{x \in \X}\sum_{r \in \N_0}  r \ell(\phi)P(x)\theta(x) (1-\E{\theta(X)})^r\\
&=\frac{\sum_{x \in \X}\ell(x)P(x)\theta(x)}{\E{\theta(X)}}+\frac{\ell(\phi)(1-\E{\theta(X)})}{\E{\theta(X)}}\\
&=\frac{\E{L(\theta)}}{\E{\theta(X)}}.
}

For $\E{R_2}$, it follows from \eqref{e:Rk} that
\[
\E{R_2} = \frac 12 \E{Y_2^2} + \E{Y_2Z_1} -\frac12 \E{Y_2},
\]
since $\E{Z_2} = \E{Z_1}$. Also, note that $Z_1$ only depends on the symbol $x\in \X$ received at time $S_1$ which in turn can depend only on the symbols $X_n$ for $n \leq S_1-1$. 
On the other hand, $Y_2=S_2-S_1$ depends on symbols $X_n$ for $n\geq S_1$ and the outputs of the independent coin tosses corresponding to randomization $\theta$. 
Therefore, $Z_1$
is independent of $Y_2$, whereby
\[
\E{R_2} = \frac 12 \E{Y_2^2} + \E{Y_2}\left(\E{Z_1} -\frac12\right).
\]
Next, note that $Z_1$ takes the value $\ell(x)$, $x\in \X$, when the
symbol received at $S_1$ is $x$. This latter event happens with
probability 
\[
\sum_{r=0}^\infty (1- \E{\theta(X)})^r\theta(x)P(x) = \frac{\theta(x)P(x)}{\E{\theta(X)}},
\]
and so, by the definition of $L(\theta)$ in  \eqref{e:L(theta)},
\begin{align*}
\E{Z_1}&= \frac{\sum_x \ell(x)\theta(x)P(x)}{\E{\theta(X)}}
\\
&=\frac{\E{L(\theta)}}{\E{\theta(X)}}  - \frac{\ell(\emptyset)(1-\E{\theta(X)})}{\E{\theta(X)}}.
\end{align*}
Then by denoting $p_\emptyset=1-\E{\theta(X)}$, the second moment $\E{Y_1^2}$ can be computed by observing the following recursion:
\begin{align*}
&\E{Y_1^2} 
\\
&=\sum_{x \in \mathcal{X}}\sum_{r \in
  \N_0}(\ell(x)+r\ell(\emptyset))^2P(x)\theta(x) p_{\emptyset}^r 
\\ 
\nonumber 
&=\sum_{x \in \mathcal{X}}\ell(x)^2P(x)\theta(x)\\ \nonumber & \hspace{1cm} +p_{\emptyset}\sum_{x \in \mathcal{X}} \sum_{r \in \N} (\ell(x)+r\ell(\emptyset))^2P(x)\theta(x) p_{\emptyset}^{r-1}
\\ 
\nonumber 
&=\sum_{x \in \mathcal{X}}\ell(x)^2P(x)\theta(x)\\ \nonumber & \hspace{1cm}+p_{\emptyset}\sum_{x \in \mathcal{X}} \sum_{r \in \N} \big(\ell(x)+(r-1)\ell(\emptyset)\big)^2P(x)\theta(x) p_{\emptyset}^{r-1}
\\
&\hspace{1cm}+2\ell(\emptyset)p_{\emptyset}\sum_{x \in \mathcal{X}} \sum_{r \in \N}
 \big(\ell(x)+(r-1)\ell(\emptyset)\big)P(x)\theta(x) p_{\emptyset}^{r-1}
\\ 
&\hspace{1cm}+p_{\emptyset}\sum_{x \in \mathcal{X}} \sum_{r \in \N} 
\ell(\emptyset)^2P(x)\theta(x) p_{\emptyset}^{r-1}
\\
\nonumber &= \sum_{x \in \mathcal{X}}\ell(x)^2P(x)\theta(x)\\ \nonumber &\hspace{1cm}  +p_{\emptyset} \E{Y_1^2}
+2\ell(\emptyset)(1-\E{\theta(X)} )\E{Y_1}+\ell(\emptyset)^2p_{\emptyset},
\end{align*}
which upon rearrangement yields
\[
\E{Y_1^2} = \frac{\E{L(\theta)^2}}{\E{\theta(X)}} 
+2\E{Y_1}\cdot \frac{\ell(\emptyset)p_{\emptyset}}{\E{\theta(X)}}.
\]
Upon combining the relations derived above, we get
\[
\frac{\E{R_2}}{\E{Y_1}}= \frac{\E{L(\theta)^2}}{2\E{L(\theta)}}
+\frac{\E{L(\theta)}}{\E{\theta(X)}}-\frac 12,
\]
which with \eqref{eq:avgage_RY} completes the proof. 

It remains to establish \eqref{eq:avgage_RY}. The proof is a simple
extension of the renewal reward theorem to our sequence of rewards
$R_n$ in which adjacent terms\newer{ may be} dependent. We include it here for
completeness. Note that $(Y_n)_{n \in
  \N}$ is a sequence of non-negative iid random variables with mean
$\E{Y_1}$, and  $S_n=\sum_{k =1}^{n}Y_k$ for all $n\in \N$. The sequence
 $\{S_n\}$ serves as a sequence of renewal times  and $R_n$ denotes the reward accumulated in the $n$th
renewal interval (though not in the
standard iid sense). Define
$N(t)$ to be the number of receptions up to time $t>0$, $i.e.$,
\eq{
N(t)=\sup{\{n: S_n\leq t\}},
} 
and $R(t)$ to be the cumulative reward accumulated till time $t$,
$i.e.$, 
\eq{
R(t)=\sum_{k=1}^{N(t)}R_k.
}
With this notation, we have
\begin{align}
\frac{R(t)}{t}&=\frac{\sum_{k=1}^{N(t)}R_k }{t}
\\
&=\frac{\sum_{k=1}^{N(t)}R_k }{N(t)}.\frac{N(t)}{t}.
\label{e:avgage_factors}
\end{align}
Note that 
\eq{
\frac{\sum_{k=1}^{\floor{\frac{N(t)}{2}}}\sum_{i \in \{0,1\}
  }R_{2k+i}}{N(t)} & \leq \frac{\sum_{k=2}^{N(t)}R_k }{N(t)} \\& \leq
\frac{\sum_{k=1}^{\ceil{\frac{N(t)}{2}}}\sum_{i \in \{0,1\}
  }R_{2k+i}}{N(t)}.
}
We now analyze the two bounds in the previous set of inequalities. 
Since $\E{Y_1}$ is finite, we get (see \cite{ross1996stochastic} for a
proof) 
\begin{align}
\lim_{t \to \infty}\frac{N(t)}{t}\to \frac{1}{\E{Y_1}}
\quad a.s.,
\label{e:Nt_limit}
\end{align}
which also shows that $N(t) \to \infty \quad a.s.$ as $t
\to \infty$. 
Therefore, for $i\in\{0,1\}$,
\eq{
\frac{\sum_{k=1}^{\ceil{\frac{N(t)}{2}}}R_{2k+i}}{N(t)}
=\frac{\sum_{k=1}^{\ceil{\frac{N(t)}{2}}}R_{2k+i}}{\ceil{\frac{N(t)}{2}}}\cdot\frac{\ceil{\frac{N(t)}{2}}}{N(t)}  . 
}
Since $(R_{2k+i})_{k \in \N}$ is iid and $N(t) \to \infty
\quad a.s.$ as $t \to \infty$, strong law of large
numbers yields 
\eq{
\lim_{t \to 
    \infty}\frac{\sum_{k=1}^{\ceil{\frac{N(t)}{2}}}R_{2k+i}}{\ceil{\frac{N(t)}{2}}}= 
  \E{R_2} \quad a.s.  \quad \forall i \in \{0,1\},
}  
which further gives
\eq{
\lim_{t \to
    \infty}\frac{\sum_{k=1}^{\ceil{\frac{N(t)}{2}}}\sum_{i \in \{0,1\}
    }R_{2k+i}}{N(t)} =\E{R_2} \quad a.s.  .
}  
Similarly,
\eq{
\lim_{t \to
    \infty}\frac{\sum_{k=1}^{\floor{\frac{N(t)}{2}}}\sum_{i \in
      \{0,1\} }R_{2k+i}}{N(t)} =\E{R_2} \quad a.s.  
.}  
Combining the observations above, we get 
\eq{
\lim_{t \to \infty}\frac{\sum_{k=1}^{N(t)}R_k }{N(t)}
  =\E{R_2} \quad a.s.  
,}
which together with \eqref{e:avgage_factors} and \eqref{e:Nt_limit} yields \eqref{eq:avgage_RY}.

\subsection{Proof of Theorem~\ref{t:update_optimal}}
Our proof is based on noticing that the minmax cost $\Delta^*(P)$ in
\eqref{e:minmax_cost}  
involves weighted average length with weights $g_{z, Q, P}(x)$. In
fact, we will see below that there is no loss in restricting to
nonnegative weights, whereby our cost has a form of average length
with respect to a distribution that depends on $(z,Q)$. 
The broad idea of the proof is to establish that a optimal code corresponding
to the {\it least favorable} choice of $(z,Q)$ is minmax
optimal. However, the proof is technical since our cost function may not satisfy the assumptions in a standard
saddle-point theorem.

A simpler form of the minmax cost $\Delta^*(P)$ from \eqref{e:varZ} is given by
\begin{align}\label{eq:Delta_simple}
{\Delta^*(P)=\min_{\ell \in \Lambda}\max_{z \geq 0} f(\ell,z),} 
\end{align}
 where 
\begin{align}\label{e:f}
f(\ell,z):=-z^2 \frac{\E{L}}{2}+z\sqrt{\E{L^2}}+\E{L}.
\end{align}
We seek to apply the following version of  Sion's minmax theorem to the function $f$.
\begin{thm}[Sion's Minmax Theorem~\cite{sion1958general}]\label{t:Sion's t}
Let $\X$ be convex space and $\Y$ be a convex, compact space. Let $h$ be a function on $\X\times\Y$ which is convex on $\X$ for every fixed $y$ in $\Y$
and concave on $\Y$ for every fixed $x$ in $\X$. Then, 
$$\inf_{x \in \X} \sup_{y \in \Y} h(x,y)= \sup_{y \in \Y} \inf_{x \in \X} h(x,y).$$
\end{thm}
Indeed, the following lemma shows that our function $f$ satisfies the convexity requirements of Sion's minmax theorem. 
\begin{lem}\label{l:conv_conc_f}
 $f(\ell,z)$ is convex in $\ell$ for every fixed $z \geq 0$
  and concave in $z$ for a fixed $\ell \in \Lambda$.
\end{lem}
\begin{proof}
To show that $f(\ell,z)$ is a convex function of $\ell$ for every
fixed $z \geq 0$, it suffices to show that $\sqrt{\E{L^2}}$ is
convex in $L=\ell(X)$. To that end, let $L_1=\ell_1(X)$ and
$L_2=\ell_2(X)$, for some $\ell_1$ and $\ell_2$ in $\lambda$. For all
$\lambda \in [0,1]$,
\[
\sqrt{ \E{\left(\lambda L_1 + (1-\lambda)L_2 \right)^2}} \leq  \lambda \sqrt{\E{ L_1^2 }} + (1-\lambda) \sqrt{\E{ L_2^2 }},
\]
where the inequality is by  Minkowski inequality for $\norm{L}_2$.

The \newer{concavity} in $z$ can be seen easily by noticing that \newer{$\frac{\partial ^2f(\ell, z)}{\partial  z^2} \leq 0$} for all $\ell$ in $\lambda$.
\end{proof}

However, our underlying domains of optimization are not compact. Our proof below circumvents this difficulty by showing that we may replace one of the domains by a compact set. For ease of reading, we divide the proof into 3 steps; we begin by
summarize the flow at a high-level.  The first step is to show that this minmax cost remains unchanged
when the domain of $z$ is restricted to a bounded interval $[0,K]$ for
a sufficiently large $K$. 
This will allow us to interchange $\min_{l \in \Lambda}$ and $\max_{z\in [0,K]}$ in the second step by using Theorem~\ref{t:Sion's t} to obtain
\begin{align}
\Delta^*(P)=\max_{z \in [0,K]} \min_{\ell \in \Lambda}f(\ell,z). 
\label{e:restricted_range_z}
\end{align}
Furthermore, we then use Corollary~\ref{c:2variational_formula} to linearize the cost. But this brings in the maximization over an additional parameter $Q$, 
which we again interchange with the minimum over $\ell$ using Sion's minmax theorem (Theorem~\ref{t:Sion's t}). Note that the required convexity of the cost function is easy to see; we note it in the following lemma.

\begin{lem}\label{l:conv_conv_g}
For every fixed $z \geq 0$, $\sum_{x \in \X}g_{z,Q,P}(x)\ell(x)$ is convex in $\ell$ for a fixed $Q \ll P$
  and concave in $Q$ for a fixed $\ell \in \Lambda$.
 \end{lem}
 \begin{proof}
For every fixed $z \geq 0$, the cost function $\sum_{x \in \X}g_{z,Q,P}(x)\ell(x)$ is linear, and thereby convex, 
in $\ell$ for a fixed $Q$ . For concavity in $Q$, note that 
for a fixed $\ell \in \Lambda$, the function $\sqrt{Q(x)}$ is a concave function of $Q(x)$, for all $x$ in $\X$. 
\end{proof}
Thus, we obtain
\[
\Delta^*(P)=\max_{z \in [0, K], Q\ll P}\min_{\ell \in \Lambda}\sum_{x\in \X}g_{z,Q,P}(x)\ell(x).
\] 
In the final step, we will establish that the optimal code
for linear cost with weights corresponding to the least favorable $(z, Q)$ is minmax optimal.
We now present each step in detail. 

\paragraph*{Step 1} We begin by noting that there is no loss in
restricting to codes with\footnote{For simplicity, we assume that
  $\log \X$ is an integer.} $\E{L}\leq \log{|\mathcal{X}|}$. Indeed, 
note that for $\E{L}>\log{|\mathcal{X}|}$ the average age is bounded as
\begin{align}
\E{L}+\frac{\E{L^2}}{2\E{L}} \geq \frac{3}{2}\E{L}> \frac 3 2\log{|\mathcal{X}|},
\label{e:age_bound_large}
\end{align}
where we have used Jensen's inequality. 
On the other hand, a
fixed-length code with $\ell(x) = \log |\X|$ attains 
\begin{align}
\E{L}+\frac{\E{L^2}}{2\E{L}} =\frac 32 \log{|\mathcal{X}|},
\label{e:fixed-length}
\end{align}
which gives the desired form
\begin{align}
\nonumber
\Delta^*(P) & = \min_{\ell \in \Lambda, \E{L}\leq
  \log{\mathcal{X}}}\E{L}+\frac{\E{L^2}}{2\E{L}}\\ & = \min_{\ell \in \Lambda, \E{L}\leq \log{\mathcal{X}} }\max_{z \in \R} f(\ell,z),
\label{e:restrict_min}
\end{align}
\newer{where $f(\ell,z)$ is defined in \eqref{e:f}}.
Also, for a fixed $\ell$ in $\Lambda$ the function $f(\ell,z)$ attains its maximum at 
$z^*(\ell)$ given by
\[
z^*(\ell):= \frac{\sqrt{\E{L^2}}}{\E{L}}.
\]
For $\E{L}\leq \log |\X|$, the maximizer $z^*(\ell)$ is bounded
as\footnote{We assume without loss of generality that $P(x)>0$ for every $x\in \X$.}
\begin{align*}
z^*(\ell) &\leq \frac{\sqrt{\E{L^2}}}{H(X)} \\ &=
\frac{\sqrt{\sum_{x}P(x)\ell(x)^2}}{H(X)} \\ &\leq
\frac{\E{L}}{H(X)}\sqrt{\max_{x\in \X}\frac 1{P(x)}} \\ &\leq
\frac{\log |\X|}{H(X)}\sqrt{\frac 1{\min_{x\in \X}P(x)}},
\end{align*}
where the first inequality uses $\E{L}\geq H(X)$, which holds for
every prefix-free code, and the second holds since $\|a\|_2 \leq
\|a\|_1$ for any sequence $a= (a_1, ..., a_n)$. 
Denoting
\[
K := \frac{\log |\X|}{H(X)}\sqrt{\frac 1{\min_{x\in \X}P(x)}},
\]
\eqref{e:restrict_min} yields \eq{ \Delta^*(P)&=\min_{\ell \in
    \Lambda, \E{L}\leq \log{|\X| }}\max_{z \in [0, K]} f(\ell, z).  }
Next, we show that the minmax cost above remains unchanged when we
drop the constraint $\E{L}\leq \log |\X|$ in the outer
minimum,
 which will complete the first step
of the proof and establish \eqref{e:restricted_range_z}.
Indeed, since by \eqref{e:fixed-length} the minimum over $\ell\in
\Lambda$ such that $\E{L}\leq \log|\X|$ is at most $(3/2)\log |\X|$, it suffices to show that 
\begin{align}
\min_{\ell \in \Lambda, \E{L} > \log |\X|}\max_{z \in [0, K]} f(\ell,
z) > \frac 32 \log |\X|.
\label{e:lb_larger_EL}
\end{align}
We divide the proof of this fact into two cases. 
First consider the case when $\ell$ in $\Lambda$ is such that  $\E{L}> \log |\X|$ and $K \geq
z^*(\ell)$. Then,  $\max_{z\in [0,K]}f(\ell,z)$ equals $\max_{z \geq 0}f(\ell,z)$, which is bounded below by $(3/2)\log |\X|$ using
\eqref{e:age_bound_large} and the definition of $f(\ell,z)$.
For the second case when $\E{L}> \log |\X|$ and $K < z^*(\ell)$, we
have
\eq{
\max_{z \in[0,K]}f(\ell,z)
&=-K^2\frac{\E{L}}{2}+K\sqrt{\E{L^2}}+\E{L}
\\
&>K^2\frac{\E{L}}{2}  +\E{L}
\\
&>\frac 32\cdot \E{L} 
\\
&>\frac{3}{2}\cdot\log|\X|,} 
where the first inequality uses $K< z^*(\ell)=\sqrt{\E{L^2}}/{\E{L}}$
and the second holds since $K\geq 1$ from its definition. Therefore, we have established~\eqref{e:lb_larger_EL}, and so we have
\eq{
\Delta^*(P)
=\min_{\ell \in \Lambda, \E{L}\leq \log{|\X| }}\max_{z \in [0, K]} f(\ell, z)
=\min_{\ell \in \Lambda}\max_{z \in [0, K]} f(\ell, z).
}

\paragraph*{Step 2} 
By lemma \ref{l:conv_conc_f} , $f(\ell,z)$ is convex in $\ell$ for every fixed $z \geq 0$
  and concave in $z$ for a fixed $\ell \in \Lambda$, $z$ takes values in a convex compact set $[0,K]$, and 
the set  $\{\ell: \ell \in \Lambda\}$ is convex, we get from Sion's minmax theorem (Theorem~\ref{t:Sion's t}) that 
\eq{
\Delta^*(P)
=\min_{\ell \in
      \Lambda}\max_{z \in [0, K]} f(\ell,z)= \max_{z \in [0,
        K]}\min_{\ell \in \Lambda}f(\ell,z).
} 
Using Corollary~\ref{c:2variational_formula}, we have
 \eq{\|L\|_2 =\max_{Q\ll P}\sum_{x\in \mathcal{X}}Q(x)^{\frac{1}{2}}
    P(x)^{\frac{1}{2}}\ell(x),
}
which by the definition of $f$ in  \eqref{e:f} further yields
\begin{align}\label{eq:Delta_3}
{f(\ell,z)=\max_{Q\ll P}\sum_{x\in \X}g_{z,Q,P}(x)\ell(x),}
\end{align}
where
\begin{align*}
g_{z,Q,P}(x) = \left(1- \frac{z^2}2\right)P(x) +z\sqrt{Q(x)P(x)}.
\end{align*}
We have obtained 
 \begin{align}
 \Delta^*(P)=\max_{z \in [0, K]}\min_{\ell
   \in \Lambda}\max_{Q\ll P}\sum_{x\in \X}g_{z,Q,P}(x)\ell(x).
 \end{align}
From Lemma~\ref{l:conv_conv_g}, $\sum_{x \in \X}g_{z,Q,P}(x)\ell(x)$ is convex in $\ell$ , for all $Q \ll P$,
  and concave in $Q$, for a fixed $\ell \in \Lambda$.
 Furthermore, since the set $\{Q:Q\ll P\}$ is convex compact for a pmf $P$ on finite alphabet, using Sion's minmax
theorem (Theorem~\ref{t:Sion's t}) once again, we get 
\begin{align}
& \Delta^*(P)=\max_{z \in [0, K]}\max_{Q\ll P}\min_{\ell \in
    \Lambda}\sum_{x\in \X}g_{z,Q,P}(x)\ell(x),
\label{eq:Delta_4}
\end{align}
which completes our second step.

 \paragraph*{Step 3} By \eqref{eq:Delta_4}, we get
\begin{align}
& \Delta^*(P)\leq\max_{z \geq 0}\max_{Q\ll P}\min_{\ell \in
    \Lambda}\sum_{x\in \X}g_{z,Q,P}(x)\ell(x).
\nonumber
\end{align}
On the other hand, by \eqref{eq:Delta_simple} and \eqref{eq:Delta_3}  we have 
\begin{align} \nonumber
\Delta^*(P)& =\min_{\ell \in
    \Lambda}\max_{z \geq 0}\max_{Q\ll P}\sum_{x\in \X}g_{z,Q,P}(x)\ell(x)\\& \geq\max_{z \geq 0}\max_{Q\ll P}\min_{\ell \in
    \Lambda}\sum_{x\in \X}g_{z,Q,P}(x)\ell(x),
\nonumber
\end{align}
whereby
\begin{align}
\nonumber
 \Delta^*(P) &=\min_{\ell \in
    \Lambda}\max_{z \geq 0}\max_{Q\ll P}\sum_{x\in \X}g_{z,Q,P}(x)\ell(x) \\ &=\max_{z \geq 0}\max_{Q\ll P}\min_{\ell \in
    \Lambda}\sum_{x\in \X}g_{z,Q,P}(x)\ell(x),
\label{e:max_min}
\end{align}
which proves the first part of theorem \ref{t:update_optimal}.

Next, we claim that in the maxmin formula above, the maximum is attained by a $(z,Q)$ for which $g_{z,Q,P}(x)$ is non-negative for every $x$. Indeed, if for some $z,Q$ there exists an $x^\prime$ in $\X$ such that $g_{z,Q,P}(x^\prime)$ is negative, then the cost $\sum_{x\in \X}g_{z,Q,P}(x)\ell(x)$ is minimized by any $\ell$ such that $\ell(x^\prime)=\infty$ and the minimum value is $-\infty$. Such $z,Q$ clearly can't be the optimizer of the maxmin problem, 
since for $z=0$, we have $g_{z,Q,P} \ge 0$, which in turn leads to $\min_{\ell \in
    \Lambda}\sum_{x\in \X}g_{z,Q,P}(x)\ell(x) \geq 0$.

Finally, consider  $(z, Q)$ such that $g_{z,Q,P}(x)\geq 0$ for all $x \in \X$. For such a $(z, Q)$, we seek to identify the minimized $\ell$ below:
\begin{align}\label{eq:Delta_5}
&\min_{\ell \in \Lambda}\sum_{x\in \X}g_{z,Q,P}(x)\ell(x) \nonumber \\&=\sum_{x'\in
  \X}g_{z,Q,P}(x')\min_{\ell \in \Lambda} \sum_{x\in
  \X}\frac{g_{z,Q,P}(x)}{\sum_{x'\in \X}g_{z,Q,P}(x')}\ell(x).
\end{align}
Thus, our optimization problem reduces to the standard problem of designing minimum average length prefix-free codes for the pmf 
\[
P_{z,Q}(x) = \frac{g_{z,Q,P}(x)}{\sum_{x'\in \X}g_{z,Q,P}(x')}.
\]
By Shannon's source coding theorem for variable length codes, the minimum is achieved by 
\[
\ell^*_{z,Q}(x):=\log\newer{ \frac{\sum_{x^{\prime}\in
      \X}g_{z,Q,P}(x)}{g_{z,Q,P}(x)}.}
\] 
Furthermore, $\ell^*_{z,Q}$ is the unique minimizer in $\Lambda$. 

Consider now a maximizer $(z^*, Q^*)$ of the maxmin problem in \eqref{e:max_min}, and let $\ell^o = \ell^*_{z^*, Q^*}$. 
Then, by Lemma~\ref{l:unique_spl} in the appendix,$(\ell^o, (z^*, Q^*))$ is a saddle-point for the minmax problem in \eqref{e:max_min}. 
  Moreover, $\ell^o$ is the unique minmax optimal solution.

\subsection{Proof of Theorem~\ref{t:qmaintheorem}}\label{p:3}
Denoting 
\begin{align}\label{e:f2_def}
f( \ell,z)= -z^2 \frac{(L_{\tt
      th}-\E{L})}{2}+z\sqrt{\E{L^2}}+\E{L},
\end{align}     
the optimal cost $\Delta^*(P)$ can be written as 
\begin{align*}
\Delta^*(P)&= \inf_{\ell \in \Lambda , \E{L} < L_{\tt
    th}}\frac{\E{L^2}}{2(L_{\tt th}-\E{L})}+\E{L}\\&
= \min_{\ell \in \Lambda, \E{L} < L_{\tt
    th}}\max_{z \geq 0} f(\ell,z).
\end{align*}
This form is similar to the one we had in Theorem~~\ref{t:update_optimal}. But the proof there does not extend to the case at  hand. 
Specifically, note that for each $\ell$, $f(\ell,z)$ attains its maximum value for $z^*(\ell)=\frac{\sqrt{\E{L^2}}}{(L_{\tt th}-\E{L})}$ which, unlike the quantity that we obtained in the proof of Theorem~\ref{t:update_optimal}, is unbounded over the set of $\ell\in \Lambda$ such that $\E{L}\leq L_{\tt th}$. However, under the additional assumption $H(X)+\log (1+1/\sqrt{2})< L_{\tt th}$, we can provide a simpler alternative proof. 
We rely on the following lemma.
\begin{lem}\label{l:minmax_conditions}
Consider a function $h:\X\times \Y \to \R$ such that the set $\X$ is compact convex, the set $\Y$ is convex, $h(x, y)$ is a convex function of $x$ for every fixed $y$ and a concave function of $y$ for every fixed $x$. Suppose additionally that there exist a convex subset $\X_0$ of $\X$ and a compact convex subset $\Y_0$ of $\Y$ such that
\begin{enumerate}
\item for every for every $x\in \X_0$, an optimizer $y^*(x)\in \arg\max_{y\in \Y}h(x,y)$ belongs to $\Y_0$; and 
\item for every $y\in \Y_0$, an optimizer $x^*(y)\in \arg\min_{x\in \X} h(x,y)$ belongs to $\X_0$.
\end{enumerate}
Then, 
\[
\min_{x\in \X}\max_{y\in \Y}h(x,y)= \max_{y\in \Y}\min_{x\in \X}h(x,y).
\]
\end{lem}
\begin{proof}
Note that since for $x$ in $\X_0$, the $y$ that maximizes $h(x,y)$ over $\Y$ is in $\Y_0$, we get 
\[
\min_{x\in \X}\max_{y\in \Y}h(x,y)\leq
\min_{x\in \X_0}\max_{y\in \Y}h(x,y)
=\min_{x\in \X_0}\max_{y\in \Y_0}h(x,y).
\]
Further, by Sion's minmax theorem (Theorem~\ref{t:Sion's t}), the right-side equals
$\max_{y\in \Y_0}\min_{x\in \X_0}h(x,y)$. But by our second assumption, the restriction $x\in \X_0$ can be dropped, and we have
\[
\max_{y\in \Y_0}\min_{x\in \X_0}h(x,y)=
\max_{y\in \Y_0}\min_{x\in \X}h(x,y)\leq \max_{y\in \Y}\min_{x\in \X}h(x,y).
\]
Thus, we have shown $\min_{x\in \X}\max_{y\in \Y}h(x,y)\leq \max_{y\in \Y}\min_{x\in \X}h(x,y)$, which completes the proof since the inequality in the other direction holds as well.
\end{proof}
For our minmax cost, we will verify that both the conditions of the lemma above hold under the assumption $H(X)+\log (1+1/\sqrt{2}) <L_{\tt th}$. Indeed, first note that for any fixed $\ell\in \Lambda$ with $\E{L}\leq H(X)+\log(1+1/\sqrt{2})$, the maximizer $z$ of $f(\ell,z)$  given by $\sqrt{\E{L^2}}/(L_{\tt th} - \E{L})$ satisfies
\begin{align*}
&\frac{\sqrt{\E{L^2}}}{L_{\tt th} - \E{L}}\\&\leq 
\sqrt{\frac 1{\min_x P(x)}}\cdot \frac{\E{L}}{L_{\tt th} - \E{L}}
\\
&\leq \sqrt{\frac 1{\min_x P(x)}}\cdot \frac{H(X)+\log (1+1/\sqrt{2})}{
L_{\tt th} - H(X)-\log (1+1/\sqrt{2})}.
\end{align*}
Denote the right-side above by $K$ and $L_{\tt th}' = H(X)+\log (1+1/\sqrt{2})$. 
Therefore, with the set $\{\ell\in \Lambda, \E{L}\leq L_{\tt th}'\}$ in the role of $\X_0$ in Lemma~\ref{l:minmax_conditions}, 
the set $[0,K]$ can play the role of $\Y_0$.

To apply Lemma~\ref{l:minmax_conditions}, we require two
conditions to hold: first, that 
$f(l,z)$ is a convex function of $\ell$ for every fixed $z$ and a concave function of $z$ for every fixed $\ell$, second, that
for every $z\in [0,K]$, the minimizing $\ell$ 
satisfies $\E{L}\leq L_{\tt th} ^ \prime$.The first easily follows from \eqref{e:f2_def}. The proof of this fact is exactly the same as Lemma~\ref{l:conv_conc_f}. However, while the second condition can be
shown to be true, the proof of
this fact is almost the same as the proof of our theorem.  For
simplicity of presentation, we instead present an alternative proof of
the theorem that uses a slight extension of the lemma above.
Note that from our foregoing discussion and following the proof of the
lemma, 
we already have obtained
\[
\Delta^*(P)\leq \max_{z \in[0,K]}\min_{\ell \in \Lambda, \E{L} \leq {L^{\prime}}_{\tt th}}f(\ell,z).
\]
By using Corollary~\ref{c:2variational_formula}  and using Sion's minmax theorem once again, we get
\begin{align*}
  \lefteqn{\Delta^*(P)}
\\
  &
  \leq \max_{z \in [0, K]}\max_{Q\ll P}\min_{\ell \in \Lambda, \E{L} \leq {L^{\prime}}_{\tt th}}\sum_{x\in \X}  g_{z,Q,P}(x)\ell(x)- \frac{z^2}{2} L_{\tt th},
\end{align*}
where 
\begin{align*}
g_{z,Q,P}(x) := \left(1+ \frac{z^2}2\right)P(x)+z\sqrt{Q(x)P(x)}.
\end{align*}    
In the preceding argument, we can use Sion's minmax theorem as the following two conditions hold. First, for every fixed $z \geq 0$, the function $\sum_{x \in \X}  g_{z,Q,P}(x)\ell(x)- \frac{z^2}{2} L_{\tt th}$ is concave in $Q$ for a fixed $\ell \in \Lambda$ and convex in $\ell$ for a fixed $Q \ll P$. Second, the sets $\{Q: Q\ll P\}$ and $\{\ell \in \Lambda: \E{L} \leq {L^{\prime}}_{\tt th}\}$ are compact and convex. Proof of the first is exactly the same as  that of \ref{l:conv_conv_g}. Second is true as we have restricted to a finite alphabet $\X$.
Thus, we can proceed as in the proof of the lemma, but  we need to show now that for every $z\in [0,K]$ and $Q\ll P$, the optimal $\ell^*(z,Q)$ satisfies $\E{L^*}\leq L_{\tt th}^{\prime}$
. 
Indeed, consider the following optimization problem for a fixed  $z$, $Q$:
\begin{align*}
\min_{\ell \in \Lambda}& \sum_{x\in \X}g_{z,Q,P}(x)\ell(x) \\&=\left(\sum_{x'\in
  \X}\newer{g_{z,Q,P}(x^{\prime})}\right) \min_{\ell \in \Lambda} \sum_{x\in
  \X}\frac{g_{z,Q,P}(x)}{\sum_{x'\in \X}g_{z,Q,P}(x')}\ell(x).
\end{align*}
Since $\frac{g_{z,Q,P}(x)}{\sum_{x'\in \X}g_{z,Q,P}(x')}$ are nonnegative and add to $1$,
in the optimization problem above, we are minimizing the expected prefix
free lengths for a finite alphabet for a particular distribution. Thus,
by Shannon's Source Coding Theorem, the optimal $\ell^*_{z,Q}$ is given by
\newer{
\begin{align*}
\ell^*_{z,Q}(x):=\log \frac{\sum_{x^{\prime}\in\X}g_{z,Q,P}(x^{\prime})}
{g_{z,Q,P}(x)};
\end{align*}}
in fact, this optimizer is unique. But then for  every $x$ in $\X$, \newer{
\eq{
&\ell^*_{z,Q}(x)\\&=\log \frac{\sum_{x^{\prime}\in \X}g_{z,Q,P}(x^{\prime})}{g_{z,Q,P(x)}}\\ 
 &=\log\frac{\sum_{x^{\prime} \in \X}\left(1+\frac{z^2}2\right)P(x)+\sum_{x\in \X}z\sqrt{Q(x)P(x)}}{\left(1+\frac{z^2}2\right)P(x)+z\sqrt{Q(x)P(x)}}
  \\
& \leq \log
  \frac{1}{P(x)} \\ & ~ +\log\left(\frac{\left(1+ \frac{z^2}2\right)}{\left(1+
    \frac{z^2}2\right)+z\sqrt{\frac{Q(x)}{P(x)}}}
  +\frac{z}{\left(1+
    \frac{z^2}2\right)+z\sqrt{\frac{Q(x)}{P(x)}}}\right)
 \\
& \leq \log
  \frac{1}{P(x)}+\log\left(\frac{\left(1+ \frac{z^2}2\right)}{\left(1+
    \frac{z^2}2\right)}
  +\frac{z}{\left(1+
    \frac{z^2}2\right)}\right)        
     \\
    &\leq \log
  \frac{1}{P(x)}+\log\left(1+\frac{1}{\sqrt{2}}\right),
   }}
where the first inequality is by the Cauchy-Schwarz inequality, the second inequality follows upon noting that $\frac{Q(x)}{P(x)}$ is nonnegative, and the last inequality follows from the fact that $z^2/2+1 \geq\sqrt{2}z$ \newer{(which holds with equality at $z=\sqrt{2}$)}. 
Thus as a consequence of this inequality the expected code length of
such a code is upper bounded as
follows, 
\begin{align}\label{e:q_corr}
\E{L^*_{z,Q}}\leq 
    H(x)+\log\left(1+\frac{1}{\sqrt{2}}\right),
\end{align}
which in the manner of Lemma~\ref{l:minmax_conditions} gives
\begin{align}
\nonumber \Delta^*(P) = \max_{z \geq 0}\max_{Q\ll P}\min_{\ell \in
  \Lambda, \E{L} \leq {L}_{\tt th}}\sum_{x\in \X} g_{z,Q,P}(x)\ell(x)-
\frac{z^2}{2} L_{\tt th}.
\end{align}
Finally, it remains to establish that $\ell^*_{z^*,Q^*}$ is the unique minmax optimal solution. This can be shown in exactly the same manner as it was shown for Theorem~\ref{t:update_optimal} in the previous section; we skip the details.\qed

\section*{Acknowledgement} We thank Rajesh Sundaresan for pointing out
connection to Campbell's work and to the
reference~\cite{SundaresanHanewal} which led to
Theorem~\ref{t:variational_formula}, the main technical contribution
of this work. Also, we thank Roy Yates for bringing to our notice references~\cite{YatesDCC, zhong2017backlog} and illuminating discussions on connection between our formulation and that in~\cite{YatesDCC}.

This work has been supported in part by the indigenous 5G Testbed
project supported by the Department of Telecommunications, Ministry of
Communications, India, 
by the Science and Engineering Research Board (SERB) under Grant
No. DSTO-1677, 
and by a PhD research grant from the Robert
Bosch Center for Cyber-Physical Systems at the Indian Institute of
Science, Bangalore.
\appendix
\section*{A saddle-point lemma}
The following simple result is needed to establish the minmax
optimality of our scheme. The first part of the result claims that any
pair of minmax optimal $x$ and maxmin optimal $y$ forms a saddle
point, a well-known fact. The second part claims that if the minimizer
for the maxmin optimal $y$ is unique, then it must also be minmax
optimal and thereby constitute a saddle-point with $y$.

\begin{lem}\label{l:unique_spl} 
Consider the minmax problem $\displaystyle{\min_{x\in \X}\max_{y\in \Y}h(x,y)}$, and
assume that
\[
\min_{x\in \X}\max_{y\in \Y}h(x,y)=\max_{y\in \Y}\min_{x\in \X}h(x,y).
\]
Then, for every pair $(x^*, y^*)$ such that $x^*\in\arg\min_{x \in
  \mathcal{X}}\max_{y \in \mathcal{Y}}h(x,y)$ and $y^*\in\arg\max_{y
  \in \mathcal{Y}}\min_{x \in \mathcal{X}}h(x,y)$  
constitutes a saddle-point. Furthermore, if the minimizer $x^o(y^*)$
of $\min_{x \in \mathcal{X}}h(x,y^*)$ is unique, then $x^*=x^o(y^*)$ 
is the unique minmax optimal solution. 
\end{lem}
\begin{proof}
Since minmax and maxmin costs are assumed to be equal, by the definition of $x^*$ and
$y^*$, we have
\begin{align}
h(x,y^*) &\geq \max_{y^\prime \in \mathcal{Y}}\min_{x^\prime \in \mathcal{X}}h(x^\prime,y^\prime)
\nonumber
\\
&=\min_{x^\prime \in \mathcal{X}}\max_{y^\prime \in \mathcal{Y}}h(x^\prime,y^\prime)\geq
h(x^*,y),
\label{e:minmax_exchange}
\end{align}
for all $x$ in $\mathcal{X}$ and $y$ in $\mathcal{Y}$.
Upon substituting $ x^*$ for $x $ and $y^*$ for $y$, we get that $x^*$
   is a minimizer of $h(x,y^*)$ and $y^*$ a maximizer of $h(x^*,
y)$. Therefore, $(x^*, y^*)$ forms a saddle-point and $h(x^*,
y^*)=\min_{x\in\X}\max_{y\in\Y}h(x,y)$.

Turning now to the second part, suppose that $x^\prime$, too, is minmax optimal. 
Then, using \eqref{e:minmax_exchange} with $x=x^\prime$ and $y=y^*$, we get that $x^\prime$ must be a minimizer
of $h(x,y^*)$ as well. 
But since this minimizer is unique,  $x^\prime$ must coincide with $x^o$.

\end{proof}
\section*{Proof of Lemma~\ref{c:dim_red}}
Denoting
\[
c_P(z,Q):=\sum_{x\in \X}g_{z,Q,P}(x) \log\frac{\sum_{x^\prime\in
    \X}g_{z,Q,P}(x^\prime)}{g_{z,Q,P}(x)},
\]
 we begin by observing the
concavity of $c_P(z,Q)$.  Recall the notations $\G=\{z \geq 0, Q \in \R^{|\X|}:g_{z,Q,P}(x)\geq0 \quad \forall x \in \X\}$
and $g_{z,Q,P}(x)  = (1-z^2/2)P(x) + z\sqrt{Q(x)P(x)}$.

\begin{lem}\label{c:concavity}
The function $c_P(z,Q) $ is concave in $Q$ for each fixed $z$
and is concave in $z$ for each fixed $Q$, over the set $\G$.
\end{lem}
\begin{proof}
For the first part, \eqref{eq:Delta_5} yields that for every $(z,Q)\in
\mathcal{G}$,
\eq{
\sum_{x\in \X}g_{z,Q,P}(x) &\log\frac{\sum_{x^\prime\in
    \X}g_{z,Q,P}(x^\prime)}{g_{z,Q,P}(x)}\\ &=\min_{\ell \in
  \Lambda}\sum_{x\in \X}g_{z,Q,P}(x)\ell(x).
}
Also, for every fixed $z$, the function $g_{z,Q,P}(x)$ is concave in $Q$, and thereby $\sum_{x\in \X}g_{z,Q,P}(x)\ell(x)$, is concave in
 $Q$. Thus, since the minimum of concave functions is concave,
$c_P(z,Q)$ is concave in $Q$ for a fixed $z$. Similarly, we can show
concavity in $z$ for a fixed $Q$ since $g_{z,Q,P}(x)$ is concave in
$z$, too, for every fixed $Q$.  
\end{proof}

We now complete the proof of Lemma~\ref{c:dim_red}.
 We will show that for any $(z,Q)$ which is feasible for optimization
 problem \eqref{e:maxmin_cost}, we can find a feasible
 $(z,Q^\prime)$ with $Q^\prime$ satisfying \eqref{e:property}, and
$$ c_P(z,Q) \leq c_P(z,Q^\prime). $$

Indeed, consider
  $Q^\prime(x):=Q(A_i)/|A_i|$ for all $x \in
  \X$.  The remainder of the proof is divided into two parts, the
first proving the feasibility of $Q^\prime$ and the second 
 proving $ c_P(z,Q) \leq c_P(z,Q^\prime). $
\paragraph{Feasibility of $(z,Q^\prime)$}  

From the feasibility of
  $(z,Q)$, for all symbols $x$ in $A_i$ and for all $i$ in
$[M_P]$, $g_{z,Q,P}(x) \geq 0$, whereby 
\begin{align*}
\sum_{x \in A_i}
    g_{z,Q,P}(x) 
&=\sum_{x \in A_i}
 \left(1- \frac{{z}^2}2\right)P(x)
  \\&\hspace{3cm}  +{z}\sum_{x \in A_i}\sqrt{Q(x)P(x)} 
\\
&= \left(1- \frac{{z}^2}2\right)P(A_i)
    +{z}\sum_{x \in A_i}\sqrt{Q(x)P(x)} 
\\
&\geq \left(1-\frac{{z}^2}2\right)P(A_i)
+{z}\sqrt{Q^\prime(A_i)P(A_i)} 
\\
&= |A_i| g_{z,Q^{\prime},P}(x) 
\\
&\geq 0,
\end{align*}
where the first inequality is by Cauchy-Schwarz inequality, the positivity
 of $z$, and the assumption that $P(x) = P(A_i)/|A_i|$ for every $x$ in
 $A_i$, and the final identity uses definition of $Q^\prime$. This
 proves the feasibility of $(z, Q^\prime)$ for the optimization
 problem \eqref{e:maxmin_cost}.

\paragraph{Proof of optimality}  
Denoting by $\Pi(A_1)$ the set of all permutations of the elements of
$A_1$, let $Q^{\pi}$ be the distribution given by
\[ 
Q^\pi(x)=\begin{cases}
&Q(\pi(x)),\quad \forall x\in A_1
\\
&Q(x), \quad \text{otherwise}.
\end{cases}
\]
Then, the distribution $\overline{Q} = (1/|\Pi(A_1)|)\cdot\sum_{\pi\in \Pi(A_1)}Q^\pi$ satisfies
\[ 
\overline{Q}(x)=\begin{cases}
&\frac 1 {|A_1|} \cdot Q(A_1),\quad \forall x\in A_1
\\
&Q(x), \quad \text{otherwise}.
\end{cases}
\]
Since by Lemma \ref{c:concavity} $c_P(z,Q)$ is concave in $Q$ for every fixed $z$, we get
\[
c_P(z,\overline{Q}) \geq \frac{1}{|\Pi(A_1)|}\cdot\sum_{\pi \in\Pi(A_1)}
c_P(z,Q^\pi).
\]
Furthermore, note that $g_{z,Q^\pi, P}(x)= g_{z, Q, P}(\pi(x))$ 
since
$P(x)=P(A_1)/|A_1|$ for every $x$ in $A_1$, and thereby 
$c_P(z, Q^\pi) = c_P(z, Q)$ for every $\pi\in \Pi(A_1)$ . Therefore,
combining the observations above, we obtain $c_P(z,\overline{Q}) \geq
c_P(z, Q)$.

Repeating this argument by iteratively using permutations of $A_i$ for $i\geq 2$, we obtain the required inequality 
\[
c_P(z, Q') \geq c_P(z,Q).
\]

\qed

\begin{IEEEbiographynophoto}
{Prathamesh Mayekar} (S'18)  received the B.E.
degree in Electronics and Telecom. Engineering from Mumbai University, India in 2013 and the  M.Tech. degree in Industrial Engineering and Operation Research from the Indian Institute of Technology Bombay, India in 2015.
Currently, he is a Ph.D. candidate at the Department of Electrical Communication Engineering, Indian Institute of Science, India. Broadly, his  research interests lie at the intersection of information theory and optimization. 
He is a  recipient of  2018 Jack Keil Wolf ISIT Student Paper Award and 
  Wipro PhD fellowship.
\end{IEEEbiographynophoto}

\begin{IEEEbiographynophoto}{Parimal Parag}
(S'04-M'11) is an assistant professor in the ECE department at Indian Institute of Science. 
Prior to that, he was a senior system engineer (R\&D) at Assia Inc. in Redwood City (2011-2014). 
He received a Ph.D. degree from the Texas A\&M University in 2011, 
the M.Tech. and B.Tech. degrees in 2004 from IIT Madras, all in electrical engineering. 
His research interests lie in the design and analysis of large scale distributed systems. 
He was a co-author of the 2018 IEEE ISIT student best paper, 
and a recipient of the 2017 early career award from the Science and Engineering Research Board.  
\end{IEEEbiographynophoto}

\begin{IEEEbiographynophoto}
{Himanshu Tyagi}
(S'04--M'14--SM'19) received the B.Tech. degree in
electrical engineering and the M.Tech. degree in
communication and information technology, both from the Indian
Institute of Technology, Delhi, India in 2007. He received 
the Ph.D. degree from 
the University of Maryland, College Park in 2013. From 2013 to 2014, he was a
postdoctoral researcher at the Information Theory and Applications (ITA) 
Center, University of California, San Diego. Since January 2015, he has been
 a faculty member at the Department of Electrical Communication Engineering, Indian Institute of Science in Bangalore. His research interests broadly lie in information theory and its application in cryptography, statistics, machine learning, and computer science. Also, he is interested in communication and automation for city-scale systems. 
\end{IEEEbiographynophoto}

\end{document}